%% file: journal_version.tex
\documentclass[journal]{IEEEtran}

\makeatletter
\def\ps@headings{%
\def\@oddhead{\mbox{}\scriptsize\rightmark \hfil \thepage}%
\def\@evenhead{\scriptsize\thepage \hfil \leftmark\mbox{}}%
\def\@oddfoot{}%
\def\@evenfoot{}}
\makeatother
\usepackage[pdftex]{graphicx}
\usepackage{amsmath}
\usepackage{amssymb}
\usepackage{array}
\usepackage[svgnames]{xcolor} 
\usepackage{mdwmath}
\usepackage{subfigure}
 \usepackage{bm}
 \usepackage{bbm}
\usepackage{tipa}
 \usepackage{hyperref}
 \usepackage{framed}
 \usepackage{booktabs}
 \usepackage{caption}
\usepackage{algorithm}
\usepackage{algorithmic}
\usepackage{epsfig}
\usepackage{overpic}
\usepackage{epstopdf}
\usepackage{amsfonts,setspace,amsthm,braket}

\graphicspath{{../pdf/}{../jpeg/}}
\DeclareGraphicsExtensions{.pdf,.jpeg,.png}

\newtheorem{theorem}{Theorem}
\newtheorem{lemma}{Lemma}

\newtheorem{corollary}{Corollary}
\newtheorem{definition}{Definition}

\newtheorem{constraint}{Constraint}

\allowdisplaybreaks

\newcommand{\conv}[1]{\operatorname{conv}\left(#1\right)}

\newcommand{\beqa}{\begin{eqnarray}}
\newcommand{\eeqa}{\end{eqnarray}}

\floatname{algorithm}{Algorithm}
\begin{document}
\title{Throughput-Optimal Multihop Broadcast on \\ Directed Acyclic Wireless Networks}
\author{
%\phantom{
%\IEEEauthorblockN{Abhishek Sinha, Georgios Paschos, 
%Chih-ping Li, and Eytan Modiano\\}
%\IEEEauthorblockA{Laboratory for Information and Decision Systems, Massachusetts Institute of Technology, Cambridge, MA 02139\\
%Email: \{sinhaa, gpasxos, cpli, modiano\}@mit.edu}
%}
\IEEEauthorblockN{Abhishek Sinha\IEEEauthorrefmark{1}, Georgios Paschos\IEEEauthorrefmark{2}, 
Chih-ping Li\IEEEauthorrefmark{3}, and Eytan Modiano\IEEEauthorrefmark{1}\\}
\IEEEauthorblockA{\IEEEauthorrefmark{1}Laboratory for Information and Decision Systems, Massachusetts Institute of Technology, Cambridge, MA 02139\\}
\IEEEauthorblockA{\IEEEauthorrefmark{2}Mathematical and Algorithmic Sciences Lab France Research Center, Huawei Technologies Co., Ltd.\\}
\IEEEauthorblockA{\IEEEauthorrefmark{3}Qualcomm Research, San Diego, CA\\
Email: \IEEEauthorrefmark{1}sinhaa@mit.edu,
\IEEEauthorrefmark{2}georgios.paschos@huawei.com,
\IEEEauthorrefmark{3}cpli@qti.qualcomm.com,
\IEEEauthorrefmark{1}modiano@mit.edu
}
\thanks{ Part of the paper appeared in the proceedings of INFOCOM, 2015, IEEE.}
%\thanks{ This work was supported by NSF Grant CNS-1217048,  ONR Grant N00014-12-1-0064, and  ARO MURI Grant W911NF-08-1-0238}
\thanks{ This work was supported by NSF Grant CNS-1217048,  ONR Grant N00014-12-1-0064, and  ARO MURI Grant W911NF-08-1-0238}
\thanks{\IEEEauthorrefmark{2} The work of G. Paschos was done while he was at MIT and affiliated with CERTH-ITI, and it was supported in part by the WiNC project of the Action: Supporting Postdoctoral Researchers, funded by national and Community funds (European Social Fund).}
\thanks{\IEEEauthorrefmark{3} The work of C.p.Li was done when he was a Postdoctoral scholar at LIDS, MIT.}
}
\maketitle

%%%%%%%%%%%%%%%%%%%%%%%%%%%%%%%%%%%%%%%%%%%%%%%%%%%%%%%%%%%%%%%%%%%%%%%%%%%%%%%%%%%%%%%%%

\input{abstract}
\input{introduction}
\input{model}
\input{capacity}

\input{DAG_algorithm}
\input{throughput_computation_dag}
\input{extension_to_cyclic_networks}
\input{simulation}

\bibliographystyle{IEEEtran}
\bibliography{MIT_broadcast_bibliography}

\input{appendix}

\end{document}

%% file: abstract.tex
\begin{abstract}
We study the problem of efficiently broadcasting packets in multi-hop wireless networks.
%At each time slot, a network controller activates non-interfering links and forwards packets to all nodes at a common rate; the maximum rate is referred to as the broadcast capacity of the wireless network.
At each time slot the network controller activates a set of non-interfering links and forwards selected copies of packets on each activated link. A packet is considered jointly received only when \emph{all nodes} in the network have obtained a copy of it. The maximum rate of jointly received packets is referred to as the broadcast capacity of the network.
Existing policies achieve the broadcast capacity 
%in the wireless setting
by  balancing traffic over a set of spanning trees, 
which are difficult to maintain in a large and time-varying wireless network.
%whose maintenance is unsuitable for wireless.
We propose a new dynamic algorithm that achieves the broadcast capacity when the underlying network topology is a directed acyclic graph (DAG). This algorithm is decentralized, utilizes local queue-length information only and does not require the use of global topological structures such as spanning trees. The principal technical challenge inherent in the problem  is the absence of work-conservation principle due to the duplication of packets, which renders traditional queuing modelling inapplicable. We overcome this difficulty by studying relative packet deficits and imposing in-order delivery constraints to every node in the network.  
%neighbor node packet deficits (=differences in the number of distinct packets received) 
%spanning trees.
Although in-order packet delivery, in general, leads to degraded throughput in graphs containing cycles, we show that it is throughput optimal in DAGs and can be exploited to simplify the design and analysis of optimal algorithms. Our characterization leads to a polynomial time algorithm for computing the broadcast capacity of any wireless DAG under the primary interference constraints. Additionally, we propose a multiclass extension of our algorithm which can be effectively used for broadcasting in any network with arbitrary topology.
Simulation results show that the our algorithm has superior delay performance as compared to the tree-based approaches.
\end{abstract}

%% file: introduction.tex
\section{Introduction and Related Work}
Broadcast refers to the fundamental network functionality of delivering data from a source node to all other nodes. 
For efficient broadcasting, we need to use appropriate packet replication and forwarding to eliminate unnecessary packet retransmissions.
 %by  eliminating unnecessary  retransmissions of data.  
%An alternative to broadcast is to use multiple  unicast sessions from the source to each node.
%However, the  use of broadcast yields superior performance since a number of unnecessary retransmissions are avoided.
This is especially important in power-constrained wireless systems which suffer from interference and collisions.
Broadcast applications include mission-critical military communications \cite{military},  
%disseminating software updates to smartphones \cite{smartphones}, 
live video streaming \cite{lstream}, and data dissemination in sensor networks \cite{akyildiz2002}.

The design of efficient wireless broadcast algorithms faces several challenges.
Wireless channels suffer from  interference, and a broadcast policy needs to activate non-interfering links at every time slot. Wireless network topologies undergo frequent changes, so that packet forwarding decisions must be made in an adaptive fashion. Existing dynamic multicast algorithms that balance traffic over spanning trees~\cite{swati} may be used for broadcasting, since broadcast is a special case of multicast.
These algorithms, however, are not suitable for wireless networks because enumerating all spanning trees is computationally prohibitive that needs to be performed repeatedly when the network topology changes with time.

In this paper, we study the fundamental problem of throughput optimal broadcasting in wireless networks.
%which in the special case of broadcast correspond to spanning trees.
%This is problematic in wireless; 
We consider a time-slotted system. At every slot, a scheduler decides which non-interfering wireless links to activate and which set of packets to forward over the activated links, so that all nodes receive packets at a common rate. The broadcast capacity is the maximum common reception rate of distinct packets over all scheduling policies.
%We then show that a specific policy has this property for any Directed Acyclic Graph.
 To the best of our knowledge, there does not exist any capacity-achieving scheduling policy for wireless broadcast without the use of spanning trees \footnote{Note that we exclude network-coding operations throughout the paper.}. The main contribution of this paper is to design provably optimal wireless broadcast algorithms that does not use spanning trees when the underlying topology is a DAG.

We start out with considering a rich class of scheduling policies $\Pi$ that perform arbitrary link activations and packet forwarding. We define the broadcast capacity $\lambda^*$ as the maximum common rate achievable over this policy class $\Pi$. We next enforce two constraints that lead to a smaller set of policies. First, we consider the subclass of policies $\Pi^{\text{in-order}}\subset \Pi$ that enforce the in-order delivery of packets. Second, we focus on the subset of policies $\Pi^* \subset \Pi^{\text{in-order}}$ that allows the reception of a packet by a node only if all its  incoming neighbours have received the packet.
%These constraints help us derive the dynamics of the constrained system and characterize its  performance. 
It is intuitively clear that the policies in the more structured class $\Pi^*$ are easier to describe and analyze, but may yield degraded throughput performance. We show the surprising result that when the underlying network topology is a directed acyclic graph (DAG), there is a control policy $\pi^* \in \Pi^*$ that achieves the broadcast capacity. In contrast, we prove the existence of a network containing a cycle in which no control policy in the policy-space $\Pi^{\text{in-order}}$ can achieve the broadcast capacity.
%there exists a policy $\pi^* \in \Pi^*$ which achieves the broadcast capacity of the network.
 
To enable the design of the optimal broadcast policy, we establish a \emph{queue-like dynamics} for the system-state, represented by relative packet deficits. This is non-trivial for the broadcast problem because explicit queueing structure is difficult to define in the network due to packet replications. We subsequently show that, the problem of achieving the broadcast capacity reduces to finding a scheduling policy \emph{stabilizing} the system, which can be accomplished by stochastic Lyapunov drift analysis techniques~\cite{tassiulas,neely2010stochastic}.  

In this paper, we make the following contributions:
\begin{itemize}
\item We define the broadcast capacity of a wireless network and show that it is characterized by an edge-capacitated graph $\widehat {\cal G}$ that arises from optimizing the time-averages of link activations. For integral-capacitated DAGs, the broadcast capacity is determined by the minimum in-degree of the graph $\widehat {\cal G}$, which is equal to the maximal number of edge-disjoint spanning trees.
\item We design a dynamic algorithm that utilizes local queue-length information to achieve the broadcast capacity of a wireless DAG network. This algorithm does not rely on spanning trees, has small computational complexity and is suitable for mobile networks with time-varying topology. This algorithm also yields a constructive proof of a version of Edmonds' disjoint tree-packing theorem \cite{edmonds} which is generalized to wireless activations but specialized to DAG topology. 
%Since it achieves the broadcast capacity, it optimizes throughput performance on a wireless network.
\item Based on our characterization of the broadcast capacity, we derive a polynomial-time algorithm to compute the broadcast capacity of any wireless DAG under primary interference constraints. 
\item We propose a randomized multiclass extension of our algorithm, which can be effectively used to do broadcast on wireless networks with arbitrary underlying topology.
\item  We demonstrate the superior delay performance of our DAG-policy, as compared to centralized tree-based algorithm \cite{swati}, via numerical simulations. We also explore the efficiency/complexity trade-off of our proposed multiclass extension through extensive simulations.
%We  demonstrate by simulations the delay performance of our policy and compare it to prior work \cite{swati}.
\end{itemize}
 
In the literature, a simple method for wireless broadcast is to use packet flooding~\cite{sasson2003probabilistic}.
The flooding approach, however, leads to redundant transmissions and collisions, known as \emph{broadcast storm} \cite{tseng2002broadcast}.
%This paper focuses on identifying the maximum performance of broadcast in a wireless network and proposing a method that achieves maximum performance by intelligently controlling  packet transmissions.
%For efficient broadcast, several copies of the same packet must exist in the network. This  complicates the description of the network state and makes it harder to deal with packets delivered out-of-order.
In the wired domain, it has been shown that forwarding \emph{useful} packets at random is optimal for broadcast \cite{massoulie2007randomized}; this approach does not extend to the wireless setting due to interference and the need for scheduling~\cite{Towsley2008}. Broadcast on wired networks can also be done using network coding~\cite{rate,Ho2005}. However, efficient link activation under network coding remains an open problem.

%For practical purposes we further  require the broadcast  solution to adapt to changes in data traffic intensity as well as to work well in large-scale systems.
%Therefore, we seek a dynamic, distributed and highly adaptive mechanism that achieves maximum performance in the wireless environment with low computational complexity.
%Broadcast can be done  using network coding .
%Network coding is a new technique  \cite{ahlswede2000network}  which  distributedly solves  the wired multicast problem \cite{deb2005network}. 
%As above, with network coding it is not clear how to activate wireless links.
%Additionally, network coding schemes suffer from large overhead and computational complexity \cite{langberg2009recent}, increased decoding delay and legacy issues, all of which make it less attractive for the wireless application domain that is the focus of this paper.
%We note that contrary to multicast, using network coding does not increase the information-theoretic broadcast capacity \cite{Li2004}.

The rest of the paper is organized as follows. 
Section \ref{System model1} introduces the wireless network model. In Section \ref{capacity_section}, we define the broadcast capacity of a wireless network and provide a useful upper bound from a cut-set consideration. In Section \ref{sec:algorithm}, we propose a dynamic broadcast policy that achieves the broadcast capacity in a DAG. In section \ref{sec:computation}, we propose an efficient algorithm for computing the broadcast capacity of any wireless DAG under primary interference constraints. Our DAG-broadcast algorithm is extended to networks with arbitrary topology in section \ref{cyclic_extension}. Illustrative simulation results are presented in Section \ref{sec:simulations}. Finally, we conclude our paper in section \ref{sec:conclusion}.

%% file: model.tex
\section{The Wireless Network Model} \label{System model1}

We consider a time-slotted wireless network model represented by the tuple $\big(\mathcal{G}(V,E),\bm{c},\mathcal{S}\big)$, where $V$ is the set of nodes, $E$ is the set of directed links, $\bm{c} = (c_{e}, e\in E)$ is the capacity-vector of the links and $\mathcal{S}$ is the set of all feasible link-activation vectors, whose elements $\bm{s} = (s_{e}, e\in E) \in \mathcal{S}$ are binary vectors such that the links $e \in E$ with $s_{e}=1$ can be activated simultaneously. The structure of the activation set $\mathcal{S}$ depends on the underlying interference model. For example, under the \emph{primary} interference constraint (also known as \emph{node-exclusive interference constraint} \cite{joo2009greedy}), the set $\mathcal{S}$ consists of all binary vectors corresponding to matchings of the underlying graph $\mathcal{G}$~\cite{west2001introduction}, see Fig.~\ref{network}.
%As an example, see Fig.~\ref{network} for a wireless network under primary interference constraints. 
 In the case of a wired network, $\mathcal{S}$ is the set of all binary vectors since there is no interference. 
%In this paper, we consider a general wireless interference model and the activation set $\mathcal{S}$ is defined accordingly. 
In this paper we allow an arbitrary link-activation set $\mathcal{S}$, which captures different wireless interference models. 
    %Thus it is clear that $\mathcal{S}$ consists of a subset of vertices of the $|E|$ dimensional hypercube and hence, bounded.  \\
%In Fig.~\ref{network} below we provide an example of a wireless network and the feasible link activations under primary interference constraint.
Let $r\in V$ be the source node at which stochastic broadcast traffic is generated (or arrives externally). The number of packets generated at the node $r$ at slot $t$ is given by the random variable $A(t) \in \mathbb{Z}_+$, which is i.i.d. over slots with mean $\lambda$.
%At source $r$, $A(t)$ packets are generated in slot $t$. 
These packets are to be delivered efficiently to all other nodes in the network.

\begin{figure}[h!]
\centering
\subfigure[a wireless network]{
	\begin{overpic}[width=0.22\textwidth]{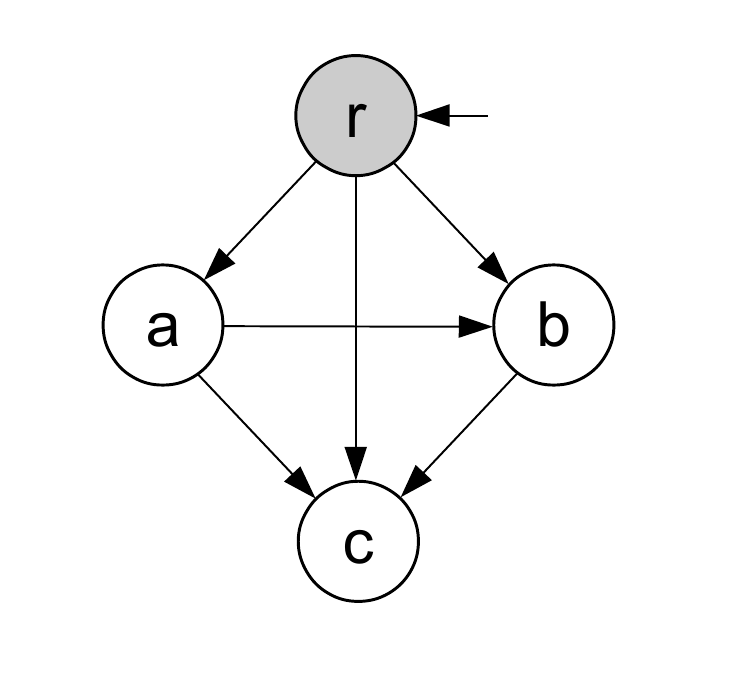}
	\put(69,77){$\lambda$}
	\end{overpic}
	\label{network-a}
}
\subfigure[activation vector $\bm{s}_1$]{
	\begin{overpic}[width=0.22\textwidth]{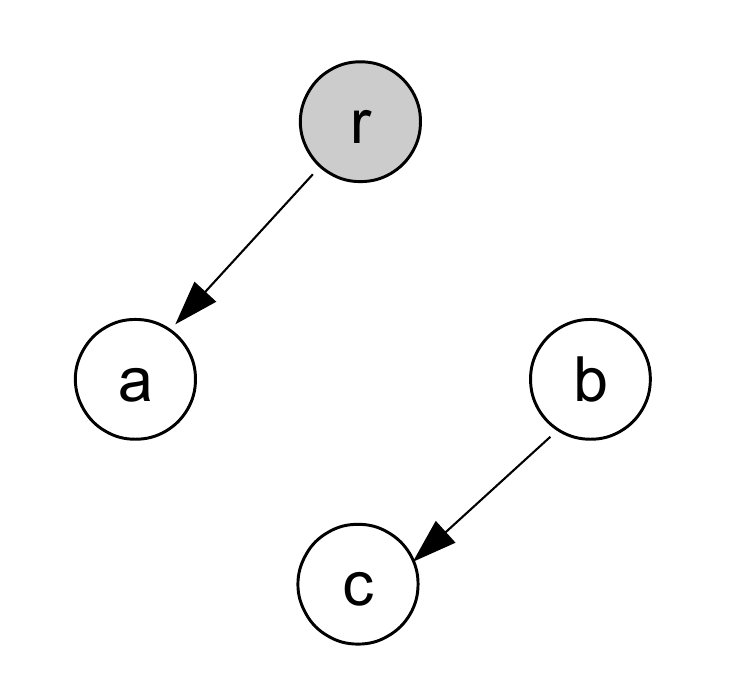}
	\end{overpic}
}
\subfigure[activation vector $\bm{s}_2$]{
  \begin{overpic}[width=0.22\textwidth]{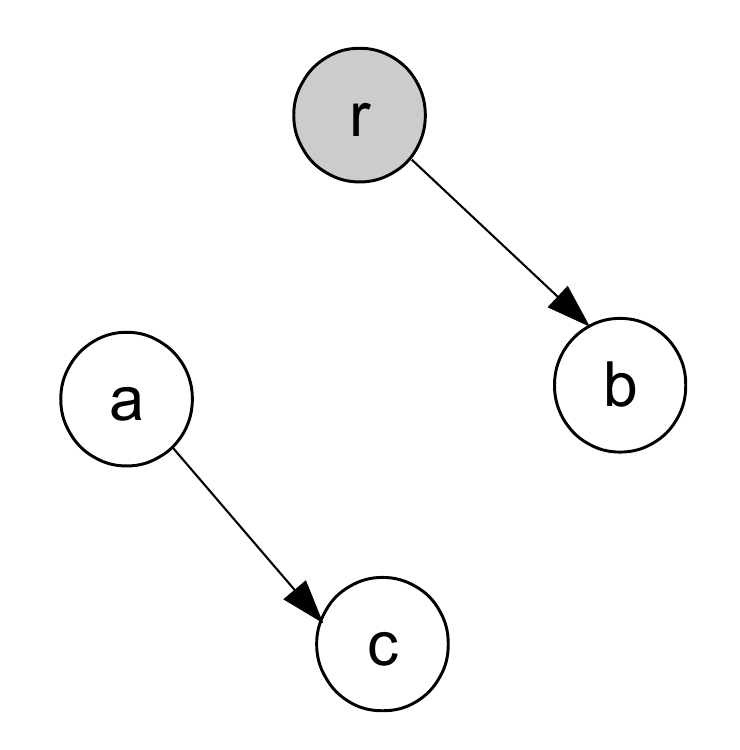}
  \end{overpic}
}
\subfigure[activation vector $\bm{s}_3$]{
  \begin{overpic}[width=0.22\textwidth]{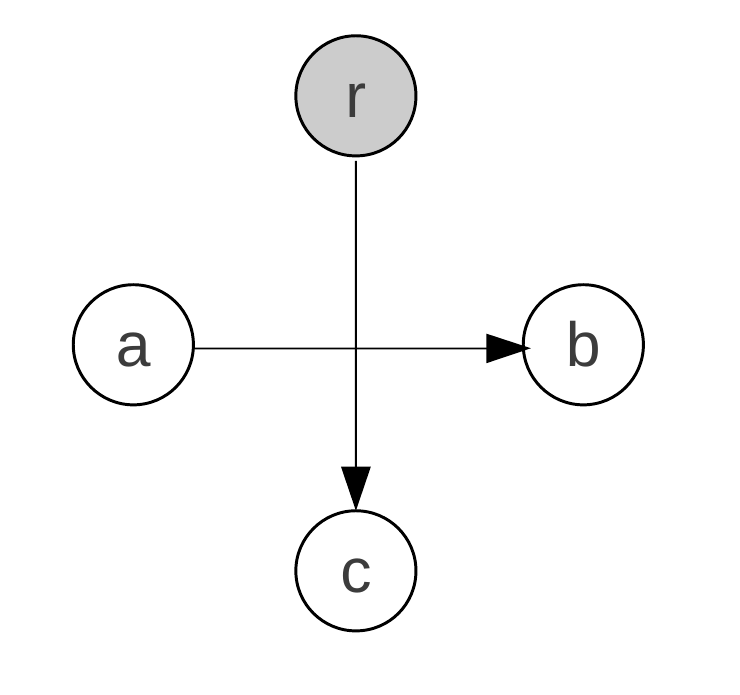}
  \end{overpic}
}
\caption{\small{A wireless network and its three feasible link activations under the primary interference constraint.}}
\label{network}
\end{figure}

%% file: capacity.tex
\section{Wireless Broadcast Capacity} \label{capacity_section}

%We start with studying the broadcast capacity of a wireless network.
 Intuitively, the network supports a broadcast rate $\lambda$ if there exists a scheduling policy under which all network nodes can receive distinct packets at  rate $\lambda$. The broadcast capacity is the maximally supportable broadcast rate in the network.
% The broadcast capacity is also related to the class of scheduling policies that activate the wireless links and forward copies of packets.
  %$\lambda$ from the source node $r$, given that the packet arrival rate at source $r$ is $\lambda$. 
Formally, we consider a class $\Pi$ of scheduling policies where each policy $\pi\in\Pi$ consists of a sequence of actions $\{\pi_t\}_{t\geq 1}$ executed at every slot $t$. Each action $\pi_{t}$ comprises of two operations: (i) the scheduler activates a subset of links by choosing a  feasible activation vector  $\boldsymbol s(t)\in\mathcal{S}$; (ii) each node $i$ forwards a subset of packets (possibly empty) to node $j$ over an activated link $(i, j)\in \mathbbm{1}(\bm{s}(t)=1)$, subject to the link capacity constraint. The class $\Pi$ includes policies that use all past and future information, and may forward any subset of packets over a link.
  
Let $R_i^{\pi}(t)$ be the number of distinct packets received by node $i \in V$ from the beginning of time up to time $t$, under a policy $\pi\in \Pi$. The time average $\liminf_{T\to \infty} R^{\pi}_i(T)/T$ is  the rate of distinct packets received at  node $i$.
%$R_i^{\pi}(t)$ will play a key role in the analysis that follows.
 
\begin{definition}
%Let the random variable $R_i^{\pi}(t)$ be the number of distinct packets received by node $i \in V$ 
%from the beginning of time up to 
%time $t$, under a policy $\pi\in \Pi$. 
A policy $\pi$ is called a 
{``broadcast policy of rate $\lambda$''} 
if 
%if and only if 
all nodes receive distinct packets at rate $\lambda$, i.e.,
\begin{equation} \label{bcdef}
\min_{i\in V} \liminf\limits_{T\to \infty} \frac{1}{T} R^{\pi}_i(T)= \lambda, \quad \text{w. p. $1$,}
\end{equation}
where $\lambda$ is the packet arrival rate at the source node $r$.
%\footnote{We can use the following more rigorous condition in~\eqref{bcdef}:
%\begin{equation} \label{eq:703}
%\min_{i\in V} \liminf_{T\to\infty} \frac{1}{T} R_{i}^{\pi}(T) = \lambda, \quad \text{w. p. $1$},
%\end{equation}
%under which all results in this paper still hold.}
%note that a rigorous definition shall use infimum instead of the regular limit. In our technical report \cite{report}, we show  that all  policies that satisfy (\ref{bcdef}) with liminf also have limits, hence the two definitions are equivalent. 
\end{definition}

\begin{definition} \label{capacity_def}
The broadcast capacity $\lambda^*$ of a wireless network is  the supremum of all arrival rates $\lambda$ for which there exists a broadcast policy $\pi \in \Pi$ of rate $\lambda$.
\end{definition}

\subsection{An upper bound on broadcast capacity $\lambda^*$} \label{broadcast_ub_section}
We characterize the broadcast capacity $\lambda^*$ of a wireless network by proving a useful upper bound. This upper bound is understood as a necessary cut-set bound of an associated edge-capacitated  graph that  reflects the time-averaged behaviour of the wireless network. We first give an intuitive explanation of the bound, assuming that the involved limits exist. Then in the proof of Theorem \ref{broadcast_ub} we rigorously prove the bound by relaxing this assumption.
%To ease exposition, we will first perform an intuitive analysis where it is assumed that all limits under consideration exist \emph{almost surely}. Although this is not necessarily true for the class $\Pi$, the proof of Theorem \ref{broadcast_ub} will establish that indeed consideration of policies without the above constraint does not alter throughput optimality.
%We provide an intuitive explanation of the bound that will be formalized in Theorem \ref{broadcast_ub} as follows. 

Fix a policy $\pi\in\Pi$. Let $\beta_e^{\pi}$ be the fraction of time link $e\in E$ is activated under $\pi$; that is, we define the vector
\begin{equation} \label{eq:101}
\bm{\beta}^{\pi} = (\beta_e^{\pi}, e\in E) = \lim_{T\to \infty}\frac{1}{T}\sum_{t=1}^{T}\bm{s}^{\pi}(t),
\end{equation}
where $\bm{s}^{\pi}(t)$ is the link-activation vector under policy $\pi$ in slot $t$. 
The average flow rate over a link $e$ under the policy $\pi$ is upper bounded by the product of the link capacity and the fraction of time the link $e$ is activated, i.e., $c_{e} \beta_{e}^{\pi}$.
%It follows that the average flow rate over a link $e$ the policy $\pi$ can utilize is at most $c_{e} \beta_{e}^{\pi}$ by the link capacity constraint. 
Hence, we can define an edge-capacitated graph 
%$\widehat{\mathcal{G}}=(V, E, \widehat{\bm{c}})$ 
$\widehat{\mathcal{G}}^\pi=(V, E,(\widehat{c}_{e}))$ 
associated with policy $\pi$, where each directed link $e\in E$ has capacity $\widehat{c}_{e} = c_{e} \beta_{e}^{\pi}$; see Fig.~\ref{cut_figure} for an example of such an edge-capacitated graph.
Next, we provide a bound on the broadcast capacity by  maximizing the broadcast capacity on the ensemble of graphs $\widehat {\cal G}^\pi$ over all feasible vectors $\bm{\beta}^{\pi}$.

We define a \emph{proper cut} $U$ of the network graph $\widehat{\mathcal{G}}^\pi$ as a proper subset of the node set $V$ that contains  the source node $r$. Define the link subset
\begin{equation} \label{eq:604}
E_{U} = \{(i, j)\in E \mid i\in U, \ j\notin U\}.
\end{equation}
%Each proper cut $U$ is associated with a cut vector $\bm{u} = (u_{ij})$ defined by
%\begin{equation} \label{eq:106}
%u_{ij} = \begin{cases} c_{ij} & \text{if } (i,j) \in E : i\in U, j \notin U, \\ 0 &\text{otherwise.} \end{cases}
%\end{equation}
Since $U \subset V$, there exists a node $n\in V\setminus U$. 
 %$U$ always separates the source node $r$ and some other node $n$ in the network. 
 Consider the throughput of node $n$ under policy $\pi$. The max-flow min-cut theorem shows that the throughput of node $n$  cannot exceed the total link capacity $\sum_{e\in E_{U}} c_{e} \, \beta_{e}^{\pi}$ across the cut $U$. Since the achievable broadcast rate $\lambda^{\pi}$ of policy $\pi$ is an upper-bound on the throughput of all nodes, we have $\lambda^{\pi} \leq \sum_{e\in E_{U}} c_{e} \, \beta_{e}^{\pi}$. This inequality holds for all proper cuts $U$ and we have
\begin{equation} \label{eq:102}
\lambda^{\pi} \leq \min_{\text{$U$: a proper cut}}\, \sum_{e\in E_{U}} c_{e} \, \beta_{e}^{\pi}.
\end{equation}
Equation~\eqref{eq:102} holds for any policy $\pi\in\Pi$. Thus, the broadcast capacity $\lambda^{*}$ of the wireless network satisfies
\begin{align*}
\lambda^{*} = \sup_{\pi\in\Pi} \lambda^{\pi} &\leq \sup_{\pi\in\Pi} \min_{\text{$U$: a proper cut}}\, \sum_{e\in E_{U}} c_{e} \, \beta_{e}^{\pi} \\ 
&\leq \max_{\bm{\beta} \in \conv{\mathcal{S}}} \min_{\text{$U$: a proper cut}}\, \sum_{e\in E_{U}} c_{e}\, \beta_{e},
\end{align*}
where the last inequality holds because the vector $\bm{\beta}^{\pi}$ associated with any policy $\pi\in\Pi$ lies in the convex hull of the activation set $\mathcal{S}$. 
Our first theorem formalizes the above intuitive characterization of the broadcast capacity $\lambda^*$ of a wireless network.
\begin{theorem} \label{broadcast_ub}
The broadcast capacity $\lambda^{*}$ of a wireless network $\mathcal{G}(V,E,\bm{c})$ with activation set $\mathcal{S}$ is upper bounded as follows:
\begin{equation} \label{eq:602}
\lambda^{*} \leq \max_{\bm{\beta} \in \conv{\mathcal{S}}} \bigg(\min_{\text{\emph{$U$: a proper cut}}}\, \sum_{e\in E_{U}} c_{e}\, \beta_{e}\bigg).
\end{equation}
\end{theorem}
\begin{IEEEproof}[Proof of Theorem~\ref{broadcast_ub}]
See Appendix~\ref{broadcast_ub_proof}.
\end{IEEEproof}

%\textbf{*The proof of Theorem~\ref{broadcast_ub} involves a cut-set bound. 
%Hence, it is possible to show that Theorem~\ref{broadcast_ub} remains true even if we allow network coding operations.
%This can be leveraged to show that the use of network coding does not increase broadcast performance in our setting, similar to  \cite{Li2004}.*}

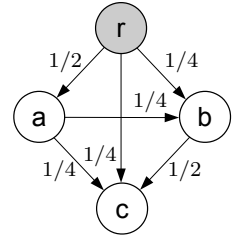
\begin{figure} [t!] 
\centering
\begin{overpic}[width=0.23\textwidth]{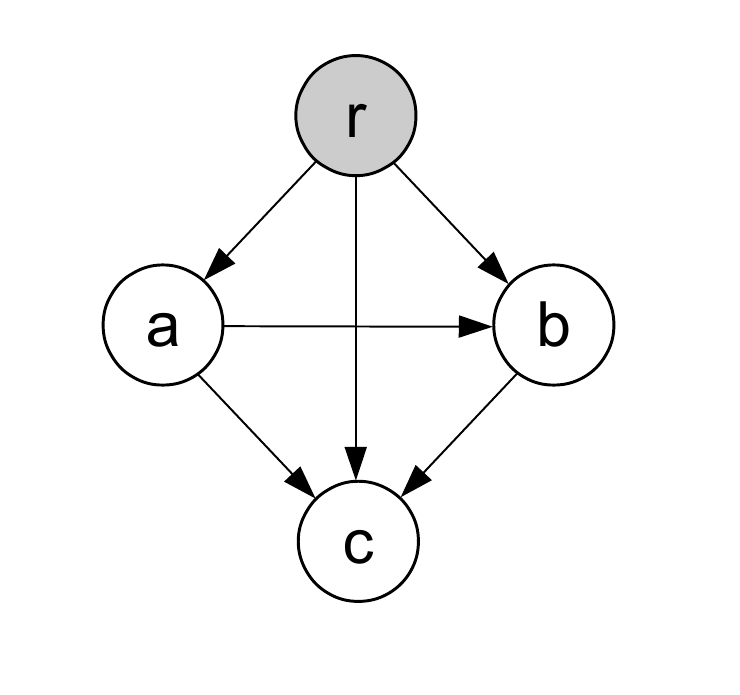}
  \put(25,66){\footnotesize $1/2$}
  \put(63,31){\footnotesize $1/2$}
  \put(52,53){\footnotesize $1/4$}
  \put(23,31){\footnotesize $1/4$}
    \put(62,66){\footnotesize $1/4$}
      \put(36,36){\footnotesize $1/4$}
  \end{overpic}
  \caption{\small{The edge-capacitated graph $\widehat{\mathcal{G}}^\pi$ for the wireless network with unit link capacities in Fig.~\ref{network} and under the time-average vector $\boldsymbol\beta^{\pi}=(1/2,1/4,1/4)$. The link weights are the capacities $c_{e} \beta_{e}^{\pi}$. The minimum proper cut in this graph has value $1/2$ (when $U = \{r,a,c\}$ or $\{r,b,c\}$). An upper bound on the  broadcast capacity is obtained by  maximizing this value over all vectors $\boldsymbol\beta^{\pi} \in\conv{\mathcal{S}}$.}}
  \label{cut_figure}
\end{figure}

\subsection{In-order packet delivery} \label{constraint1}

Studying the performance of any arbitrary broadcast policy $\pi \in \Pi$ is formidable because packets are replicated across the network and may be received out of order. 
To avoid unnecessary re-transmissions, the nodes must keep track of the identity of the received set of packets, which complicates the system state; because instead of the number of packets received, the system state is properly described here by the subset of packets received at each of the nodes.
%, \emph{resulting in an exponentially large network state space that prohibits the design and analysis of an efficient broadcast policy. Specifically, if we denote with $A(t)$  the number of packets generated at the source node $r$ by time $t$, then the network state in slot $t$ is represented by $\bm{S}(t) = (S_i(t), i\in V)$, where $S_i(t)$ is the collection of packets received by node $i$ up to time $t$ and takes values in the $2^{A(t)}$ possible subsets of the packets $ \{1,2,\ldots, A(t)\}$.}

To simplify the system state, we focus on the subset of policies $\Pi^{\text{in-order}}\subset \Pi$ that enforce the following constraint: 
\begin{constraint}[In-order packet delivery]\label{con:1}
A network node is allowed to receive a packet $p$  only if all packets $\{1, 2, \ldots, p-1\}$ have been received by that node. 
\end{constraint}
%\textbf{We could relax this constraint within a time slot.}
 In-order packet delivery is useful in live media streaming applications~\cite{lstream}, where buffering out-of-order packets incurs increased delay that degrades video quality. In-order packet delivery greatly simplifies the network state space.   
%The set of distinct packets received at a node are now completely described by the largest ID of the received packets. 
Let $R_i(t)$ be the number of distinct packets received by node $i$ by time $t$. 
 For policies in $\Pi^{\text{in-order}}$,  the set of received packets by time $t$ at node $i$ is $\{1,\dots, R_i(t)\}$.  
  Therefore, the network state in slot $t$ is given by the vector $\bm{R}(t)=\big(R_i(t), i\in V\big)$.
  %  Under any policy $\pi \in \Pi_{\text{in-order}}$, the network state in slot $t$ is completely represented by the vector $\bm{R}(t)=\big(R_i(t), i\in V\big)$, where $R_i(t)$ is the largest ID of the received packets by node $i$ by time $t$. Due to in-order packet delivery, $R_{i}(t)$ is also the total number of distinct packets received by node $i$, as in definition \ref{bcdef}.

In section \ref{sec:algorithm} we will prove that there exists a throughput-optimal broadcast policy in the space $\Pi^{\text{in-order}}$ when the underlying network topology is a DAG. Ironically, Lemma \eqref{in_order} shows that there exists a network containing a cycle in which any broadcast policy in the space $\Pi^{\text{in-order}}$ is \emph{not} throughput optimal. Hence the space $\Pi^{\text{in-order}}$ can not, in general, be extended beyond DAGs while preserving throughput optimality.  
\begin{lemma} \label{in_order}
Let ${\lambda^*}_{\text{in-order}}$ be the broadcast capacity of the policy subclass $\Pi^{\text{in-order}} \subset \Pi$ that enforces in-order packet delivery. There exists a network topology containing a directed cycle such that ${\lambda^*}_{\text{in-order}} < \lambda^{*}$.
\end{lemma}
\begin{IEEEproof}[Proof of Lemma~\ref{in_order}]
See Appendix \ref{in_order_proof}.
\end{IEEEproof}
We will return to the problem of broadcasting in networks with arbitrary topology in Section \ref{cyclic_extension}.
%Lemma~\ref{in_order} exploits the existence of a directed cycle to show that in-order packet delivery strictly reduces the broadcast capacity of the specific cyclic wireless network. 

\subsection{Achieving the broadcast capacity in a DAG} \label{sec:601}

%In the rest of the paper, we study optimal broadcast algorithms when the network topology is a DAG.
At this point we concentrate our attention to Directed Acyclic Graphs (DAGs). Graphs in this class are appealing for our analysis because they possess well-known topological ordering of the nodes \cite{west2001introduction}. For DAGs, the upper bound~\eqref{eq:602} on the broadcast capacity $\lambda^{*}$ in Theorem \ref{broadcast_ub} will be simplified further. For each receiver node $v\neq r$, consider the proper cut $U_v$ that separates the network from node $v$:
\begin{equation}\label{eq:pcutsrec}
U_v=V\setminus \{v\}.%, v\in V\setminus\{r\}
\end{equation}
Using these cuts $\{U_{v}, v\neq r\}$, we define another upper bound $\lambda_{\text{DAG}}$ on the broadcast capacity $\lambda^{*}$ as:
\begin{align}  \label{eq:DAGcap}
\lambda_{\text{DAG}}&\triangleq \max_{\bm{\beta}\in \conv{\mathcal{S}}}\min_{\{U_v, v\neq r\} }\, \sum_{e\in E_{U_{v}}} c_{e}\, \beta_{e}\\
&\geq \max_{\bm{\beta} \in \conv{\mathcal{S}}} \min_{\text{$U$: a proper cut}}\, \sum_{e\in E_{U}} c_{e}\, \beta_{e}\geq \lambda^{*}, \notag
\end{align}
where the first inequality uses the subset relation $\{U_v, v\neq r\}\subseteq \{\text{$U$: a proper cut}\}$ and the second inequality follows from Theorem \ref{broadcast_ub}. In Section \ref{sec:algorithm}, we will propose a dynamic policy that belongs to the policy class $\Pi_{\text{in-order}}$ and achieves the broadcast rate $\lambda_{\text{DAG}}$. Combining this result with~\eqref{eq:DAGcap}, we establish that the broadcast capacity of a DAG is given by
 \begin{align} 
 \lambda^*=\lambda_{\text{DAG}} &=\max_{\bm{\beta}\in \conv{\mathcal{S}}}\min_{\{U_v, v\neq r\} }\, \sum_{e\in E_{U_{v}}} c_{e}\, \beta_{e}, \label{eq:603} \\
 &= \max_{\bm{\beta} \in \conv{\mathcal{S}}} \min_{\text{$U$: a proper cut}}\, \sum_{e\in E_{U}} c_{e}\, \beta_{e}. 
 \nonumber \end{align}
This is achieved by a broadcast policy that uses in-order packet delivery. In other words, we show that imposing the in-order packet delivery constraint does not reduce the broadcast capacity when the underlying topology is a DAG.\\
From a computational point of view, the equality in Eqn. \eqref{eq:603} is attractive, because it implies that for computing the broadcast capacity of any wireless DAG, it is enough to consider only those cuts that separate a single (non-source) node from the source-side. Note that, there are only $|V|-1$ of such cuts, in contrast with the total number of cuts, which is exponential in the size of the network.  This fact will be exploited in section \ref{sec:computation} to develop a strongly poly-time algorithm for computing the broadcast capacity of any DAG under the primary interference constraints.

%% file: DAG_algorithm.tex
\section{DAG Broadcast Algorithm}\label{sec:algorithm}

In this section we design an optimal broadcast policy for wireless DAGs. We start by imposing an additional constraint that leads to a new subclass of policies $\Pi^*\subseteq \Pi^{\text{in-order}}$. As we will see, policies in $\Pi^*$ can be described in terms of relative packet deficits which constitute a simple dynamics.
%As we will see, $\Pi^*$ admits a simple network dynamics by means of relative packet deficits.
We analyze the dynamics of the minimum relative packet deficit at each node~$j$, where the minimization is over all incoming neighbours of $j$. This quantity plays the role of virtual queues in the system and we design a dynamic control policy that stabilizes them. The main result of this section is to show that this control policy achieves the broadcast capacity whenever the network topology is a DAG.
%\textbf{by creating proper system states that have simple dynamics over time and provide useful metrics so that the network broadcast capacity is achieved by optimizing the metrics.}

\subsection{System-state by means of packet deficits}

%We have shown that the network state can be represented by the vector $\bm{R}(t) = (R_{i}(t), i\in V)$ under the in-order packet delivery constraint, where $R_{i}(t)$ is the total number of packets received by node $i$ by time $t$. 
We showed in Section~\ref{constraint1} that, constrained to the policy-space $\Pi^{\text{in-order}}$, the system-state is completely represented by the vector $\bm{R}(t)$.
%We have shown that the network state can be represented by the vector $\bm{R}(t) = (R_{i}(t), i\in V)$ under the in-order packet delivery constraint, where $R_{i}(t)$ is the total number of packets received by node $i$ by time $t$. 
To simplify the system dynamics further, we restrict $\Pi^{\text{in-order}}$ further as follows. \\
\indent We say that node $i$ is an \emph{in-neighbor} of node $j$ iff there exists a directed link $(i, j)\in E$ in the underlying graph $\mathcal{G}$.
% space $\Pi_{\text{in-order}}$. 
%We consider the subset of policies in $\Pi_{\text{in-order}}$ such that
\begin{constraint}\label{con:2}
A packet $p$ is eligible for  transmission to node $j$ at a slot $t$ only if all the in-neighbours of  $j$ have received packet $p$ in some previous slot.
 \end{constraint}

%in other words, a packet $p$ cannot be transmitted to node $j$ unless all the in-neighbors of node $j$ have received the packet. 
We denote this new policy-class by $\Pi^{*} \subseteq \Pi^{\text{in-order}}$. We will soon show that it contains an optimal policy. Fig.~\ref{policy_fig} shows the relationship among different policy classes\footnote{We note that, if the network contains a directed cycle, then a deadlock might occur under a policy in $\Pi^{*}$ and may yield zero broadcast throughput. However, this problem does not arise when the underlying topology is a DAG.}.\\
%See Fig.~\ref{pi*_figure} for an example of the operations of a policy in $\Pi^{*}$.
%An efficient data broadcast policy shall observe the system state $\bm{R}(t)$ and make link activation and packet forwarding decisions in every slot $t$. Providing an explicit expression of how $R_{i}(t)$ is updated over time, however, is an important but challenging task. To come up with the proper system dynamics based on which an optimal broadcast policy can be designed, 
%\begin{figure} [h!] 
%\centering
%\begin{overpic}[width=0.39\textwidth]{pi_star_policy}
%  \put(11,28){\footnotesize $R_a(t)=18$}
%  \put(27,8){\footnotesize $R_b(t)=15$}
%  \put(52,8){\footnotesize $R_c(t)=14$}
%  \put(88,17){\footnotesize $R_j(t)=10$}
%  \end{overpic}
%  \caption{\small Under a policy $\pi\in \Pi^*$, the set of packets available for transmission to node $j$ in slot $t$ is $\{11,12,13,14\}$,
%  which are available at all in-neighbors of node $j$.
%  %since packets $\{15,16,17, 18\}$ are not yet received by the in-neighbor $c$.
%  The in-neighbor of $j$ inducing the smallest packet deficit is $i^*_t=c$, and $X_{j}(t) = \min\{Q_{aj}(t), Q_{bj}(t), Q_{cj}(t)\} = 4$. }
%  \label{pi*_figure}
%\end{figure}
%We aim at constructing a data broadcast policy in $\Pi^{*}$ that achieves the broadcast capacity $\lambda^{*}$ in the general policy space $\Pi$ when the underlying network graph is a DAG.
Following properties of the system-states $\bm{R}(t)$ under a policy $\pi \in \Pi^{*}$ will be useful.
\begin{lemma} \label{Q_positivity_lemma}
For $j \neq r$, let $\text{In}(j)$ denote the set of in-neighbors of a node $j$ in the network. 
%Suppose that packets are indexed in the order of their arrival to the source node $r$ (i.e., packet $p$ is the $p$th packet that arrives at the source $r$). 
%Let $R_{i}(t)$ be the total number of distinct packets received by node $i$ by time $t$; node $i$ has received packets $\{1, 2, \ldots, R_{i}(t)\}$ by time $t$ due to the in-order delivery constraint. 
Under any policy $\pi\in\Pi^{*}$, we have:
\begin{enumerate}
\item[(1)] $R_{j}(t) \leq \min_{i\in \text{In}(j)} R_{i}(t)$ 
\item[(2)] The indices of packets $p$ that are eligible to be transmitted to the node $j$ at slot $t$ is given by
\[
\big\{p \mid R_{j}(t)+1 \leq p \leq \min_{i\in\text{In}(j)} R_{i}(t)\big\}.
\]
\end{enumerate}
\end{lemma}
%Fig.~\ref{policy_fig} shows the relationship between different policy classes discussed in this paper.
\begin{figure} [htbp] 
\centering
\hspace*{1.5in}
\begin{overpic}[width=0.24\textwidth]{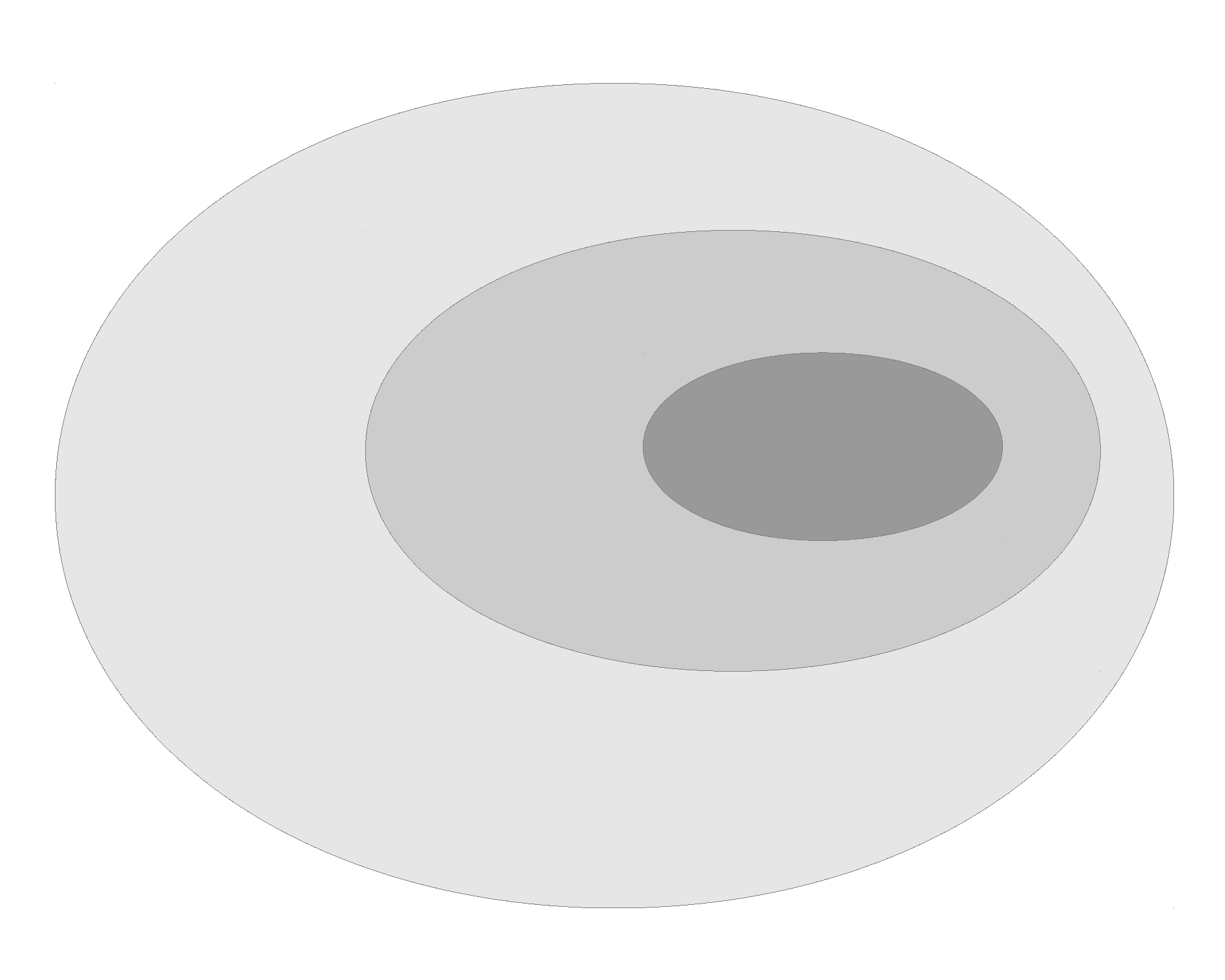}
\put(12,42){\scriptsize $\Pi$}
\put(33,42){\scriptsize  $\Pi^{\text{in-order}}$}
\put(71,42){\scriptsize  $\Pi^*$}
\put(61,38){\scriptsize  $\pi^*$}
\put(63,43){\circle*{2}}
\put(-90,64){\scriptsize  $\Pi$: all policies that perform}
\put(-81,57){\scriptsize   link activations and routing}
\put(-90,44){\scriptsize  $\Pi^{\text{in-order}}$: policies that enforce}
\put(-68,37){\scriptsize in-order packet delivery}
\put(-90,24){\scriptsize  $\Pi^*$: policies that allow reception}
\put(-78,17){\scriptsize only if all in-neighbors have}
\put(-78,10){\scriptsize received the specific packet}
 \end{overpic}
 \caption{\small Containment relationships among different policy classes.}
 \label{policy_fig}
%\end{minipage}
\end{figure}
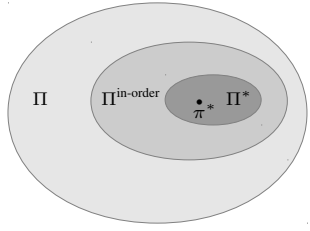

We define the \emph{packet deficit} over a directed link $(i,j)\in E$ by $Q_{ij}(t) = R_i(t)-R_j(t)$. Under a policy in $\Pi^*$, $Q_{ij}(t)$ is always non-negative because, by part (1) of Lemma~\ref{Q_positivity_lemma}, we have
\[
Q_{ij}(t) = R_{i}(t) - R_{j}(t) \geq \min_{k\in\text{In}(j)} R_{k}(t) - R_{j}(t) \geq 0.
\]
The quantity $Q_{ij}(t)$ denotes the number of packets received by node $i$ but not by node~$j$, upto time $t$.
%;  $Q_{ij}(t)$ can be interpreted as an indicator of packet deficiencies at node $j$ as compared to the upstream node $i$. 
Intuitively, if all packet deficits $Q_{ij}(t)$ are bounded asymptotically, the total number of packets received by any node is not lagging far from the total number of packets generated at the source; hence, the broadcast throughput will be equal to the packet generation rate.
%Consider a path $\widehat{e} = (e_{1}, \ldots, e_{k}, \ldots)$ as a sequence of directed links $e_{k} \in E$ from the source node $r$ to a network node $j$ (where the ingress node of the first link $e_{1}$ is the source node $r$ and the egress node of the last link in $\widehat{e}$ is node $j$). Then, the quantity $R_{r}(t) - R_{j}(t) = \sum_{e \in\widehat{e}} Q_{e}(t)$ is the number of packets that have arrived at the source $r$ but yet to be received by node $j$. Roughly speaking, if all state variables $\{Q_{e}(t), e\in E\}$ are kept bounded over time, then the throughput of each node $j$, which is $\lim_{t\to\infty} R_{j}(t)/t$, satisfies
%\begin{equation} \label{eq:103}
%\lim_{t\to\infty} \frac{R_{j}(t)}{t} = \lim_{t\to\infty} \frac{R_{r}(t)}{t} - \lim_{t\to\infty} \sum_{e\in\widehat{e}} \frac{Q_{e}(t)}{t} = \lambda - 0 = \lambda.
%\end{equation}
%Equation~\eqref{eq:103} shows that we can achieve the broadcast throughput $\lambda$ by stabilizing the ``virtual queues'' $\{Q_{e}(t), e\in E\}$, given that the exogenous data arrival rate $\lambda$ is supportable.

To analyze  the system dynamics under a policy in $\Pi^{*}$, it is useful to define the \emph{minimum packet deficit} at node $j\neq r$ by
\begin{equation} \label{X_def}
X_j(t) = \min_{i\in \text{In}(j)} Q_{ij}(t).
\end{equation}
From part~(2) of Lemma~\ref{Q_positivity_lemma}, $X_j(t)$ is the maximum number of packets that node $j$ is allowed to receive from its in-neighbors at slot $t$. As an example, Fig.~\ref{pi*_figure} shows that the packet deficits at node $j$, relative to the upstream nodes $a$, $b$, and $c$, are $Q_{aj}(t)=8$, $Q_{bj}(t)=5$, and $Q_{cj}(t)=4$, respectively. Thus $X_{j}(t)=4$ and node $j$ is only allowed to receive four packets in slot $t$ due to Constraint \ref{con:2}.
%the second constraint that defines the policy space $\Pi^{*}$. 
We can rewrite $X_{j}(t)$ as
\begin{equation} \label{eq:115}
X_{j}(t) = Q_{i_{t}^{*}j}(t), \quad \text{where } i_{t}^{*} = \arg\min_{i \in \text{In}(j)} Q_{ij}(t),
\end{equation}
and the node $i_{t}^{*}$ is the in-neighbor of node $j$ from which node $j$ has the smallest packet deficit in slot $t$; ties are broken arbitrarily in deciding $i_{t}^{*}$.\footnote{We note that the minimizer $i_{t}^{*}$ is a function of the node $j$ and the time slot $t$; we slightly abuse the notation by neglecting $j$ to avoid clutter.} Our optimal broadcast policy will be described in terms of the minimum packet deficits $\{X_{j}(t)\}$.

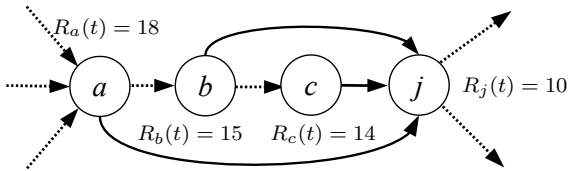
\begin{figure} [h] 
\centering
\begin{overpic}[width=0.39\textwidth]{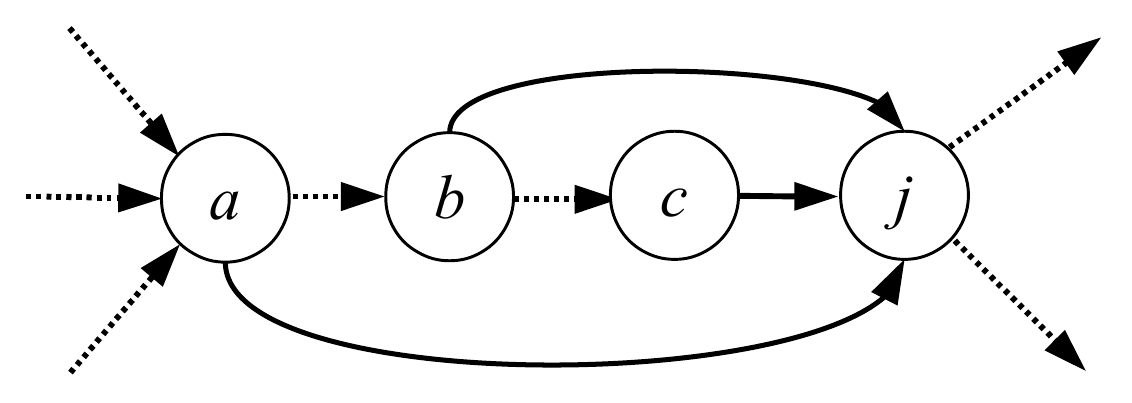}
  \put(11,28){\footnotesize $R_a(t)=18$}
  \put(27,8){\footnotesize $R_b(t)=15$}
  \put(52,8){\footnotesize $R_c(t)=14$}
  \put(88,17){\footnotesize $R_j(t)=10$}
  \end{overpic}
  \caption{\small Under a policy $\pi\in \Pi^*$, the set of packets available for transmission to node $j$ in slot $t$ is $\{11,12,13,14\}$,
  which are available at all in-neighbors of node $j$.
  %since packets $\{15,16,17, 18\}$ are not yet received by the in-neighbor $c$.
  The in-neighbor of $j$ inducing the smallest packet deficit is $i^*_t=c$, and $X_{j}(t) = \min\{Q_{aj}(t), Q_{bj}(t), Q_{cj}(t)\} = 4$. }
  \label{pi*_figure}
\end{figure}

%\begin{figure} [h!] 
%\centering
%\begin{overpic}[width=0.39\textwidth]{Figures/Drawings/Q_ev_pic}
% \put(76,16){\footnotesize $j$}
% \put(19,16){\footnotesize $i_t^*$}
% \put(47.5,16){\footnotesize $k_2$}
% \put(50,30){\footnotesize $ \mu_{k_1j}(t)$}
% \put(55,20){\footnotesize $ \mu_{k_2j}(t)$}
% \put(50,4){\footnotesize $ \mu_{i_t^*j}(t)$}
% \put(0,20.5){\footnotesize $\mu_{l_2i_t^*}$}
% \put(11,30){\footnotesize $\mu_{l_1i_t^*}$}
% \put(12.5,7){\footnotesize $ R_{i_t^*}(t)$}
% \put(72,7.5){\footnotesize $ R_{j}(t)$}
%  \end{overpic}
%  \caption{Illustration of one-step evolution of $X_{j}(t)\equiv \min_{k\in \text{In}(j)}Q_{kj}(t)\equiv R_{i_t^*(j)}(t)-R_j(t) $ for a policy in class $\Pi^*$.}
%    \label{Q_ev_fig}
%\end{figure}

\subsection{The dynamics of the system variable $X_{j}(t)$}
We now analyze the dynamics of the system variables
\begin{equation} \label{eq:104}
X_j(t) = Q_{i_{t}^{*}j}(t) = R_{i_{t}^{*}}(t)-R_{j}(t)
\end{equation}
under a policy $\pi \in \Pi^*$. 
Define the service rate vector $\bm{\mu}(t) = (\mu_{ij}(t))_{(i, j)\in E}$ by
\[
\mu_{ij}(t) = \begin{cases} c_{ij} & \text{if $(i, j)\in E$ and the link $(i, j)$ is activated,} \\ 0 & \text{otherwise.} \end{cases}
\]
Equivalently, we may write $\mu_{ij}(t) = c_{ij} s_{ij}(t)$, and the number of packets forwarded over a link is constrained by the choice of the link-activation vector $\bm{s}(t)$. At node $j$,  the increase in the value of $R_j(t)$ depends on the identity of the received packets; in particular, node $j$ must receive distinct packets. Next, we clarify which packets are to be received by node $j$ at time $t$.
 
The number of available packets for reception at node $j$ is $\min\{X_{j}(t), \sum_{k\in V} \mu_{kj}(t)\}$, because: (i) $X_{j}(t)$ is the maximum number of packets node $j$ can receive from its in-neighbours subject to the Constraint \ref{con:2}; (ii) $\sum_{k\in V} \mu_{kj}(t)$ is the total incoming transmission rate at node $j$ under a given link-activation decision. 
%See Fig.~\ref{Q_ev_fig} for an illustration. 
To correctly derive the dynamics of $R_j(t)$, we consider the following efficiency requirement on policies in $\Pi^*$:
\begin{constraint}[Efficient forwarding]\label{con:3}
Given a service rate vector $\bm{\mu}(t)$, node $j$ pulls from the activated incoming links the following subset of packets (denoted by their indices)
\begin{equation} \label{eq:701}
\Big\{p \mid R_j(t)+1\leq p\leq R_j(t)+\min\{X_{j}(t), \sum_{k\in V} \mu_{kj}(t)\}\Big\},
\end{equation}
The specific subset of packets that are pulled over each incoming link are disjoint but otherwise arbitrary.\footnote{Due to Constraints \ref{con:1} and~\ref{con:2}, the packets in~\eqref{eq:701} have been received by all in-neighbors of node $j$.}
\end{constraint}
Constraint~\ref{con:3} requires that scheduling policies must avoid forwarding the same packet to a node over two different incoming links. Under certain interference models such as the primary interference model, at most one incoming link is activated at a node in a slot and Constraint \ref{con:3} is redundant.

In Eqn.~\eqref{eq:104}, the packet deficit $Q_{i_{t}^{*}j}(t)$ increases with $R_{i_{t}^{*}}(t)$ and decreases with $R_{j}(t)$, where $R_{i_{t}^{*}}(t)$ and $R_{j}(t)$ are both non-decreasing. Hence, we can upper-bound the increment of $Q_{i_{t}^{*}j}(t)$ 
 by the total capacity $\sum_{m\in V} \mu_{mi_{t}^{*}}(t)$ of the activated incoming links at node~$i_{t}^{*}$. 
 Also, we can express the decrement of $Q_{i_{t}^{*}j}(t)$  by the exact number of distinct packets received by node $j$ from its in-neighbours, and it is given by $\min\{X_{j}(t), \sum_{k\in V} \mu_{kj}(t)\}$ by Constraint \ref{con:3}. Consequently, the one-slot evolution of the variable $Q_{i_{t}^{*}j}(t)$ is given by\footnote{We emphasize that the node $i_{t}^{*}$ is defined in~\eqref{eq:115}, depends on the particular node $j$ and time $t$, and may be different from the node $i_{t+1}^{*}$.}
\begin{align}
Q_{i_t^*j}(t+1) &\ \leq \big(Q_{i_t^*j}(t) - \sum_{k\in V} \mu_{kj}(t)\big)^+  + \sum_{m\in V}\mu_{m i_t^*}(t) \notag \\
&\ = \big(X_j(t)  - \sum_{k\in V} \mu_{kj}(t)\big)^+ + \sum_{m\in V}\mu_{mi_t^*}(t), \label{eq:105}
\end{align}
where $(x)^{+} = \max(x, 0)$ and we recall that $X_j(t)=Q_{i_t^*j}(t)$. It follows that $X_{j}(t)$ evolves over slot $t$ according to
%\begin{equation} 
\begin{align}\label{bnd2}
 X_j(t+1) &\stackrel{(a)}{=} \min_{i\in \text{In}(j)} Q_{ij}(t+1) \stackrel{(b)}{\leq} Q_{i_t^*j}(t+1) \notag\\
&\stackrel{(c)}{\leq} \big(X_j(t)  - \sum_{k\in V} \mu_{kj}(t)\big)^+ + \sum_{m\in V} \mu_{mi_t^*}(t),
\end{align}
%\end{equation}
where the equality (a) follows the definition of $X_{j}(t)$, equality (b) follows because node $i_{t}^{*} \in \text{In}(j)$ and equality (c) follows from~Eqn. \eqref{eq:105}. In Eqn. ~\eqref{bnd2}, if $i_t^*=r$, we abuse the notation to define $\sum_{m\in V} \mu_{mr}(t) = A(t)$ for the source node $r$, where $A(t)$ is the number of exogenous packet generated at slot $t$.
%With slight abuse of notation, in~\eqref{bnd2} we define $\sum_{m\in V}\mu_{mr}(t)$ as the exogenous arrivals $A(t)$ at the source node $r$ in slot $t$.

\subsection{The optimal broadcast policy} \label{lyapunov}

%Every network node can receive the throughput that is equal to a feasible packet arrival rate $\lambda$ if the minimum deficiency variables $X_{j}(t) = \min_{i\in\text{In}j)} Q_{ij}(t) = \min_{i\in\text{In}(j)} R_{i}(t) - R_{j}(t)$ for all nodes $j\neq r$ are kept bounded over time. 
Our broadcast policy is designed to keep the minimum deficit process $\bm{X}(t)$ stable.
For this, we regard the variables $X_{j}(t)$ as virtual queues that follow the dynamics~\eqref{bnd2}. By performing drift analysis on the virtual queues $X_{j}(t)$, we propose the following max-weight-type broadcast policy $\pi^*$, described in Algorithm \ref{DAG_algo}. We have $\pi^* \in \Pi^{*}$ and it enforces the constraints \ref{con:1}, \ref{con:2}, and \ref{con:3}.
 %that enforces the two requirements: (i) in-order packet delivery; (ii) a network node can only receive packets that have been possessed by all its in-neighbors. 
 We will show that this policy achieves the broadcast capacity $\lambda^{*}$ of a wireless network over the general policy class $\Pi$ when the underlying topology is a DAG.

%\noindent \rule[0.05in]{3.5in}{0.01in}

\begin{algorithm} 
\caption{Optimal Broadcast Policy $\pi^{*}$ for a Wireless DAG:}\label{DAG_algo}
At each slot $t$, the network-controller observes the state-variables $\{R_{j}(t), j\in V\}$ and executes the following actions
\begin{algorithmic}[1]
%\textbf{Optimal Broadcast Policy $\pi^{*}$ over a Wireless DAG:}
%\\ \phantom{a}
%\REQUIRE . 
%At slot $t$, note the values $R_{j}(t), j\in V$.
%\textbf{Step 1:}
\STATE For each link $(i, j)\in E$, compute the deficit $Q_{ij}(t) = R_{i}(t) - R_{j}(t)$ and the set of nodes $K_{j}(t)\subset \text{out}(j)$ for which node $j$ is their deficit minimizer, given as follows 
\begin{equation} \label{eq:110}
K_{j}(t) \gets  \big\{k\in V\mid j = \arg\min_{m\in\text{In}(k)} Q_{mk}(t)\big\}.
\end{equation}
 The ties are broken arbitrarily (e.g., in favor of the highest indexed node) in finding the $\arg\min(\cdot)$ in Eqn.\eqref{eq:110}. \\
%\textbf{Step 2:} 
 \STATE Compute $X_{j}(t) = \min_{i\in\text{In}(j)} Q_{ij}(t)$ for $j\neq r$ and assign to link $(i, j)$ the weight
\begin{equation} \label{eq:111}
W_{ij}(t) \gets \big(X_{j}(t) - \sum_{k\in K_{j}(t)} X_{k}(t)\big)^{+}.
\end{equation}
%where $W_{ij}(t)$ is the minimum deficit of node $j$ minus that of all nodes for which node $j$ is the deficit minimizer. Intuitively, the term $W_{ij}(t)$ arises because while delivering a packet to node $j$ decreases $X_{j}(t)$ by one, it also increases $X_{k}(t)$ by one for all nodes for which node $j$ is an in-neighbor and the deficit minimizer. \\
%\textbf{Step 3:} 
 \STATE In slot $t$, choose the link-activation vector $\bm{s}(t) = (s_{e}(t), e\in E)$ such that
\begin{equation} \label{eq:601}
%\bm{s}(t)  \in \arg\max_{\bm{s} = (s_{e}, e\in E)\in \mathcal{S}} \sum_{e\in E} c_{e} s_{e} W_{e}(t)
\bm{s}(t)  \in \arg\max_{ \bm{s}\in \mathcal{S}} \sum_{e\in E} c_{e} s_{e} W_{e}(t).
\end{equation}
\STATE Every node $j\neq r$ uses activated incoming links to pull  packets $\{R_j(t)+1, \dots,R_j(t)+ \min\{\sum_ic_{ij} s_{ij}(t), X_{j}(t)\}\}$ from its in-neighbors according to Constraint \ref{con:3}.\\
% Over each link $(i, j)\in E$, node $i$ transmits the next $\min\{c_{ij} s_{ij}(t), X_{j}(t)\}$ packets to node $j$ according to the in-order delivery constraint.
%\textbf{what happens if two links reach the same node? I think every node should pull packets from incoming neighbors}
%\textbf{Step 4:} 
 \STATE The vector $(R_{j}(t), j\in V)$ is updated as follows:
\[
R_{j}(t+1) \gets \begin{cases} R_{j}(t) + A(t), & j = r, \\ R_{j}(t) + \min\{\sum_ic_{ij} s_{ij}(t), X_{j}(t)\}, & j\neq r, \end{cases}
\]
%and $R_{j}(0)=0$ for all $j\in V$.
\end{algorithmic}
\end{algorithm} 
 
%\noindent \rule[0.05in]{3.5in}{0.01in}

We illustrate the Algorithm \ref{DAG_algo} in an example in Fig.~\ref{algorithm_fig}. 
\begin{figure}[t!] 
\subfigure{
\label{fig:601}
\begin{overpic}[width=0.23\textwidth]{Network_2b}
\put(41,92){\scriptsize \textbf{Step 1}}
\put(57,80){\scriptsize $R_r(t)=10$}
\put(-1,61){\scriptsize $R_a(t)=3$}
\put(67,61){\scriptsize $R_b(t)=3$}
\put(58,21){\scriptsize $R_c(t)=2$}
\put(102,81){\scriptsize $Q_{ra}(t)=7$}
\put(102,71){\scriptsize $Q_{rb}(t)=7$}
\put(102,61){\scriptsize $Q_{rc}(t)=8$}
\put(102,51){\scriptsize $Q_{ab}(t)=0$}
\put(102,41){\scriptsize $Q_{ac}(t)=1$}
\put(102,31){\scriptsize $Q_{bc}(t)=1$}
\put(142,81){\scriptsize $K_{r}(t)=\{a\}$}
\put(142,71){\scriptsize $K_{a}(t)=\{b,c\}^*$}
\put(142,61){\scriptsize $K_{b}(t)=\{\emptyset\}$}
\put(142,51){\scriptsize $K_{c}(t)=\{\emptyset\}$}
\end{overpic}
}
\subfigure{
\label{fig:602}
\begin{overpic}[width=0.23\textwidth]{Network_2b}
\put(41,92){\scriptsize \textbf{Step 2}}
%\put(57,80){\scriptsize $R_r(t)=10$}
\put(-1,61){\scriptsize $X_a(t)=7$}
\put(85,71){\scriptsize $X_b(t)=0$}
\put(58,21){\scriptsize $X_c(t)=1$}
\put(70,81){\scriptsize $W_{ra}(t)=(X_a(t)-X_b(t)-X_c(t))^+=6$}
%\put(111,71){\scriptsize $=6$}
\put(120,71){\scriptsize $W_{rb}(t)=(X_b(t))^+=0$}
\put(120,61){\scriptsize $W_{rc}(t)=(X_c(t))^+=1$}
\put(120,51){\scriptsize $W_{ab}(t)=(X_b(t))^+=0$}
\put(120,41){\scriptsize $W_{ac}(t)=(X_c(t))^+=1$}
\put(120,31){\scriptsize $W_{bc}(t)=(X_c(t))^+=1$}
\end{overpic}
}
\subfigure{
\label{fig:603}
\begin{overpic}[width=0.23\textwidth]{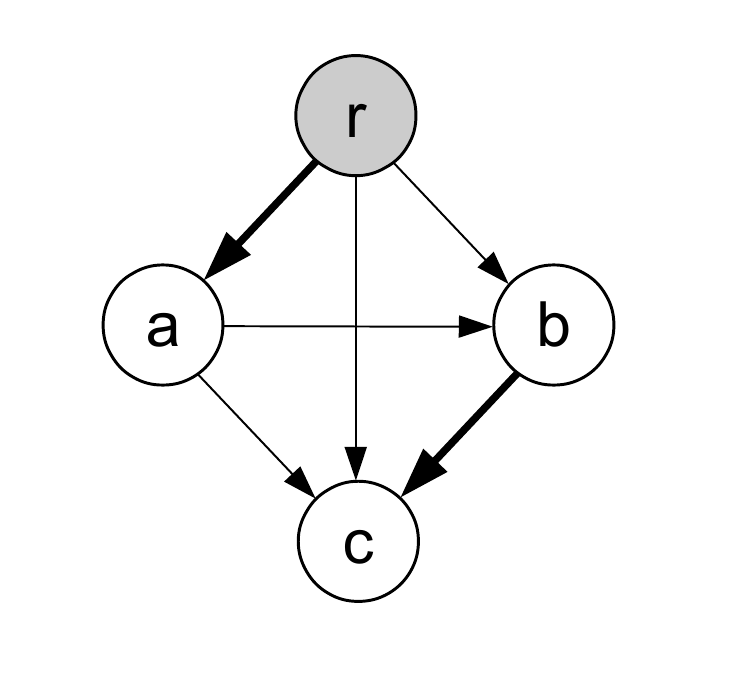}
\put(41,92){\scriptsize \textbf{Step 3}}
%\put(57,80){\scriptsize $R_r(t)=10$}
\put(-1,36){\scriptsize $R_a(t)=3$}
\put(58,16){\scriptsize $R_c(t)=2$}
\put(102,81){\scriptsize $\mathbf{s}_1$: $W_{ra}(t)+W_{bc}(t)=7$}
\put(102,71){\scriptsize $\mathbf{s}_2$: $W_{rb}(t)+W_{ac}(t)=1$}
\put(102,61){\scriptsize $\mathbf{s}_3$: $W_{rc}(t)+W_{ab}(t)=1$}
\put(102,51){\scriptsize Choose the link-activation vector $\mathbf{s}_1$}
\put(102,41){\scriptsize Forward the next packet \#4 over $(r,a)$}
\put(102,31){\scriptsize Forward the next packet \#3 over $(b,c)$}
\put(29,67){\footnotesize  \#4}
\put(63,29){\footnotesize  \#3}
\end{overpic}
}
\subfigure{
\label{fig:604}
\begin{overpic}[width=0.23\textwidth]{Network_2b}
\put(41,92){\scriptsize \textbf{Step 4}}
\put(57,80){\scriptsize $R_r(t+1)=11$}
\put(-6,65){\scriptsize $R_a(t+1)=4$}
\put(65,62){\scriptsize $R_b(t+1)=3$}
\put(58,21){\scriptsize $R_c(t+1)=3$}
\put(112,81){\scriptsize One packet arrives at the source}
 \end{overpic}
}
 \caption{\small Running the optimal broadcast policy $\pi^{*}$ in slot $t$ in a wireless network with unit-capacity links and under the primary interference constraint. Step 1: computing the deficits $Q_{ij}(t)$ and $K_{j}(t)$; a tie is broken in choosing node $a$ as the in-neighbor deficit minimizer for node $c$, hence $c\in K_{a}(t)$; node $b$ is also a deficit minimizer for node $c$. Step 2: computing $X_{j}(t)$ for $j\neq r$ and $W_{ij}(t)$. Step 3: finding the link activation vector that is a maximizer in~\eqref{eq:601} and forwarding the next in-order packets over the activated links. Step 4: a new packet arrives at the source node $r$ and the values of $\{R_{r}(t+1), R_{a}(t+1), R_{b}(t+1), R_{c}(t+1)\}$ are updated.
 }
 \label{algorithm_fig}
%\end{minipage}
\end{figure}
The next theorem demonstrates the optimality of the broadcast policy $\pi^{*}$.
\newpage 
\begin{theorem} \label{main_theorem}
If the underlying network graph $\mathcal{G}$ is a DAG, then for any exogenous packet arrival rate $\lambda <\lambda_{\text{DAG}}$, the broadcast policy $\pi^{*}$ yields
\[
\min_{i \in V} \liminf\limits_{T \to \infty}\frac{R_i^{\pi^{*}}(T)}{T} = \lambda, \quad \text{w.p. $1$,}
\]
where $\lambda_{\text{DAG}}$ is the upper bound on the broadcast capacity $\lambda^{*}$ in the general policy class $\Pi$, as shown in~\eqref{eq:DAGcap}. Consequently, the broadcast policy $\pi^{*}$ achieves the broadcast capacity $\lambda^{*}$ for any Directed Acyclic Graphs.
\end{theorem}
\begin{IEEEproof}[Proof of Theorem~\ref{main_theorem}]
See Appendix \ref{main_theorem_proof}.
\end{IEEEproof}

\subsection{Number of disjoint spanning trees in a DAG} \label{sec:disjoint}
Theorem~\ref{main_theorem} provides an interesting combinatorial result that relates the number of disjoint spanning trees in a DAG to the in-degrees of its nodes.
\begin{lemma} \label{lem:701}
Consider a directed acyclic graph $G=(V, E)$ that is rooted at a node $r$, has unit-capacity links, and possibly contains parallel edges. The maximum number $k^{*}$ of edge-disjoint spanning trees in $G$ is given by
\[
k^*=\min_{v\in V\setminus \{r\}}d_{\text{in}}(v),
\]
where $d_{\text{in}}(v)$ denotes the in-degree of the node $v$.
\end{lemma}
\begin{IEEEproof}[Proof of Lemma~\ref{lem:701}]
See Appendix~\ref{pf:701}.
\end{IEEEproof}

%% file: throughput_computation_dag.tex
\section{Efficient Algorithm for Computing the Broadcast Capacity of a DAG} \label{sec:computation}
In this section we exploit Eqn. \eqref{eq:603} and develop an LP to compute the broadcast capacity of any wireless DAG network under the primary interference constraints. Although this LP has exponentially many constraints, using a well-known separation oracle, it can be solved in strongly polynomial time via the ellipsoid algorithm \cite{bertsimas1997introduction}.\\
Under the primary interference constraint, the set of feasible activations of the graphs are \emph{matchings} \cite{west2001introduction}. For a subset of edges $E'\subset E$, let $\chi^{E'} \in \{0,1\}^{|E|}$ where $\chi^{E'}(e)=1$ if $e \in E'$ and is zero otherwise. Let us define
\begin{eqnarray} 
\mathcal{P}_\text{matching}(\mathcal{G})= \textbf{convexhull}(\{\chi^{M}|M \text{ is a matching in } G\})
\end{eqnarray}
We have the following classical result by Edmonds \cite{schrijver2003combinatorial}.
\begin{theorem}
The set $\mathcal{P}_\text{matching}(\mathcal{G})$ is characterized by the set of all $\bm{\beta} \in \mathbb{R}^{|E|}$ such that : 
\begin{eqnarray}
\beta_e &\geq& 0 \hspace{10pt} \forall e \in E  \label{match_poly}\\
\sum_{e \in \delta_{\text{in}}(v)\cup\delta_{\text{out}}(v)} \beta_e &\leq& 1 \hspace{10pt} \forall v \in V \nonumber\\
\sum_{e\in E[U]}\beta_e &\leq& \frac{|U|-1}{2}; \hspace{10pt} U\subset V, |U| \text{ odd} \nonumber 
\end{eqnarray}
Here $E[U]$ is the set of edge (ignoring their directions) with both end points in U, $\delta_\text{in}(u)$ ($\delta_\text{out}(u)$) denotes the set of all incoming (outgoing) edges to (from) the vertex $u \in V$.  
\end{theorem}

Hence following Eqn. \eqref{eq:603}, the broadcast capacity of a DAG can be obtained by the following LP :
\begin{framed}
\begin{eqnarray} \label{LP_match}
\max \lambda 
\end{eqnarray}
Subject to, 
\begin{eqnarray}
\lambda &\leq& \sum_{e\in \delta_\text{in}(v)} c_{e}\beta_{e}\hspace{10pt} \forall v \in V\setminus \{r\}\label{b_c_constraint}\\
\bm{\beta} &\in& \mathcal{P}_\text{matching}(\mathcal{G}) \label{matching_constr}
\end{eqnarray}
\end{framed}

From the equivalence of optimization and separation (via the ellipsoid method), it follows that the above LP is poly-time solvable if there exists an efficient separator oracle for the constraints \eqref{b_c_constraint}, \eqref{matching_constr}. Since there are only linearly many constraints ($|V|-1$, to be precise) in \eqref{b_c_constraint}, the above requirement reduces to an efficient separator for the matching polytope \eqref{matching_constr}. We refer to a classic result from the combinatorial-optimization literature which shows the existence of such efficient separator for the matching polytope
\begin{theorem}{\cite{schrijver2003combinatorial}}
There exists a strongly poly-time algorithm, that given $\mathcal{G}=(V,E)$  and $\bm{\beta} : E \to \mathbb{R}^{|E|}$ determines if $\bm{\beta}$ satisfies \eqref{match_poly} or outputs an inequality from \eqref{match_poly} that is violated by $\bm{\beta}$.
\end{theorem}
This directly leads to the following theorem.
\begin{theorem}
There exists a strongly poly-time algorithm to compute the broadcast capacity of any wireless DAG under the primary interference constraints.
\end{theorem}
The following corollary implies that, although there are exponentially many matchings in a DAG, to achieve the broadcast capacity, randomly activating (with appropriate probabilities) only $|E|+1$ matchings suffice. 
\begin{corollary}
The optimal broadcast capacity $\lambda^*$ in a wireless DAG, under the primary interference constraints, can be achieved by randomly activating (with positive probability) at most $|E|+1$ matchings.
\end{corollary}
\begin{proof}
Let $(\lambda^*,\bm{\beta^*})$ be an optimal solution of the LP \eqref{LP_match}. Hence we have $\bm{\beta}^* \in \mathcal{P}_\text{matching}(\mathcal{G})\equiv \textbf{convexhull}(\{\chi^{M}|M \text{ is a matching in } G\})$. Since the polytope $\mathcal{P}_\text{matching}(\mathcal{G})$ is a subset of $\mathbb{R}^{|E|}$, by Carath{\'e}odory's theorem \cite{matouvsek2002lectures}, the vector $\bm{\beta}^*$ can be expressed as a convex combination of at most $|E|+1$ vertices of the polytope $\mathcal{P}_\text{matching}(\mathcal{G})$, which are matchings of the graph $\mathcal{G}$. This concludes the proof. 
\end{proof}

%% file: extension_to_cyclic_networks.tex
\section{Broadcasting on Networks with Arbitrary Topology}  \label{cyclic_extension}
In this section we extend the broadcast policy for a DAG to networks containing cycles. From the negative result of Lemma \ref{in_order}, we know that any policy ensuring \emph{in-order} packet delivery at every node cannot, in general, achieve the broadcast capacity of a network containing cycles. To get around this difficulty, we introduce the concept of broadcasting using multiple \emph{classes} $\mathcal{K}$ of packets. The idea is as follows: each class $k \in \mathcal{K}$ has a one-to-one correspondence with a specific permutation $\prec_k $ of the nodes; for an edge $(a,b) \in E$ if the node $a$ appears prior to the node $b$ in the permutation $\prec_k$ (we denote this condition by $a\prec_k b$), then the edge $(a,b)$ is included in the class $k$, otherwise the edge $(a,b)$ ignored by the class $k$. The set of all edges included in the class $k$ is denoted by $E^k \subset E$. It is clear that each class $k$ corresponds to a unique embedded DAG topology $\mathcal{G}^k(V,E^{k})$, which is a subgraph of the underlying graph $\mathcal{G}(V,E)$. \\
A new incoming packet arriving at the source node is admitted to some class $k \in \mathcal{K}$, according to some policy. All packets in a given class $k\in \mathcal{K}$ are broadcasted while maintaining in-order delivery property within the class $k$, however packets from different classes do not need to respect this constraint. Hence the resulting policy does not belong to the class $\Pi^*$ in but rather to the general class $\Pi$. This new policy keeps the best of both worlds: (a) its description-complexity is $\Theta(kN)$, where for each class we essentially have the same representations as in the in-order delivery constrained policies and (b) by relaxing the inter-class in-order delivery constraint it has the potential to achieve the full broadcast capacity of the   underlying graph. \\
Hence the broadcast problem reduces to construction of multiple classes (which are permutations of the vertices $V$) out of the given directed graph such that it covers the graph efficiently, from a broadcast-capacity point of view. In Algorithm-\ref{multiclass_algo}, we choose the permutations uniformly at random with the condition that the source always appears at the first position of the permutation.\\
 \begin{algorithm}  
\caption{Multiclass Broadcast Algorithm for General Topology}
\begin{algorithmic}[1] 
 \REQUIRE Graph $\mathcal{G}(V,E)$, total number of classes $K$
 %\STATE Global variable : $\bm{R}(t)$, Local variables : $\bm{Q}(t),\bm{Q}^{G}(t),\mathcal{A}(t),\bm{X}(t), \bm{W}(t)$.
 \STATE Generate $K$ permutations $\{\prec_i\}_{i=1}^{K}$ of the nodes $V$ uniformly at random (with the source $\{r\}$ at the first position) and obtain the induced DAGs $G^{k}(V,E^{k})$, where $e=(a,b) \in E^{k}$ iff $a \prec_k b$.
 \STATE For each permutation $\prec_k$, maintain a class $k$ and the packet-counter variables $\{R_i^{(k)}\}$ at every node $i=1,2,\ldots,|V|$.
 \STATE Each class observes intra-class packet forwarding constraints \textbf{(1), (2)} and \textbf{(3)} described in sections \ref{capacity_section} and \ref{sec:algorithm}. 
\STATE Define the state variables $\{\bm{Q}^{k}(t), \bm{X}^{k}(t)\}$ and compute the weights $\{\bm{W}^{k}(t)\}$, for each class $k=1,2,\ldots, K$ exactly as in Eqn. \eqref{eq:111}, where each class $k$ considers the edges $E^k$ only for Eqns. \eqref{eq:110} and \eqref{eq:111}.
 \STATE An incoming packet to source $r$ at time $t$ joins the class $l$ corresponding to 
 \begin{eqnarray}
 \arg\min_{l\in \mathcal{K}} \sum_{j \in K^{l}_r(t)}X_j^{l}(t)
 \end{eqnarray}
\STATE The overall weight for an edge $e$ (taken across all the classes) is computed as 
\begin{eqnarray} \label{max_act}
W_e(t)=\max_{k: e\in E^{k}}W_e^{k}(t)
\end{eqnarray}
\STATE Activate the edges corresponding to the max-weight activation, i.e., 
\begin{eqnarray}
\bm{s}(t)  \in \arg\max_{ \bm{s}\in \mathcal{S}} \sum_{e\in E} c_{e} s_{e} W_{e}(t).
\end{eqnarray}
\STATE For each activated edge $e \in \bm{s}(t)$, forward packets corresponding to a class achieving the maximum in Eqn. \eqref{max_act}.
%\FOR
%\STATE g
%\ENDFOR
\label{multiclass_algo}
 \end{algorithmic}
 \end{algorithm} 
 
 \begin{theorem} \label{multiclass_th}
 The multiclass broadcast Algorithm-\ref{multiclass_algo} with $K$ classes supports a broadcast rate of 
 \begin{eqnarray}\label{multiclass_rate}
 \lambda^K = \max_{\sum_{k}\bm{\beta}^{k} \in \textrm{conv}({\mathcal{S}})} \sum_{k=1}^{K} \min_{j \neq r}\sum_{i}c_{ij}\beta_{ij}^{k}
 \end{eqnarray}
 where we use the convention that $\beta_{ij}^{k}=0$ if $(i,j) \notin E^{(k)}$. 
 \end{theorem}
 The right hand side of Eqn. \eqref{multiclass_rate} can be understood as follows. Consider a feasible stationary activation policy $\pi_{\text{STAT}}$ which activates class $l$ on the edge $(i,j)$ $\beta_{ij}^{k}$ fraction of time. Since, by construction, each of the class follows a DAG, lemma \eqref{lem:701} implies that the resulting averaged graph has a broadcast capacity of $\lambda^{k}= \min_{j}\sum_{i}c_{ij}\beta_{ij}^{k}$ for the class $k$. Thus the total broadcast rate achievable by this scheme is simply $\lambda^{K}=\sum_{k=1}^{K}\lambda^k=\sum_{k} \min_{j}\sum_{i}c_{ij}\beta_{ij}^{k}$. Given the $K$ classes, following the same line of argument as in \eqref{LP_match}, we can develop a similar LP to compute the broadcast capacity \eqref{multiclass_rate} of all these $k$-classes taken together in strongly poly-time.

 The proof of Theorem \eqref{multiclass_th} follows along the exact same line of argument as in Theorem \eqref{main_theorem}, where we now work with the following Lyapunov function $\hat{L}(\bm{Q}(t))$, which takes into account all $k$ classes: 
 \begin{eqnarray}
 \hat{L}(\bm{Q}(t))=\sum_{k=1}^{K} \sum_{j\neq r}(X_j^{k} (t) )^2
 \end{eqnarray}
 We then compare the multiclass broadcast algorithm \ref{multiclass_algo} with the stationary activation policy $\pi_{\text{STAT}}$ above to show that the Multiclass broadcast algorithm is stable under all arrival rates below $\lambda$. The details are omitted for brevity.\\
 Since the broadcast-rate $\lambda^K$ achievable by a collection of $K$ embedded DAGs in a graph $\mathcal{G}$ is always upper-bounded by the actual broadcast capacity $\lambda^*$ of $\mathcal{G}$, we have the following interesting combinatorial result from Theorem \eqref{multiclass_th}
 \begin{corollary}
 Consider a wired network, represented by the graph $\mathcal{G}(V,E)$. For a given integer $K \geq 1$, consider $K$ classes as in Theorem \eqref{multiclass_th}, with $\{E^{k}\}_{k=1}^{K}$ being their corresponding edge-sets. Then, for any set of non-negative vectors $\{\bm{\beta}^k\}_{k=1}^{K}$ with $\sum_k \beta_{ij}^{k} \leq 1, \forall (i,j)$, the following lower-bound for the broadcast capacity $\lambda^*$ holds: 
 \begin{eqnarray}
 \lambda^* \geq \sum_{k=1}^{K} \min_{j\neq r}\sum_{i}c_{ij}\beta_{ij}^{k}
 \end{eqnarray}
 where we use the convention that $\beta_{ij}^k=0$ if $(i,j) \notin E^k$.
 \end{corollary}
 The above corollary may be contrasted with Eqn. \eqref{eq:DAGcap}, which provides an upper bound to the broadcast capacity $\lambda^*$.

%% file: simulation.tex
\section{Simulation Results}\label{sec:simulations}

We present a number of simulation results concerning the delay performance of the optimal broadcast policy $\pi^{*}$ in wireless DAG networks with different topologies. For simplicity, we assume primary interference constraints throughout this section. Delay for a packet is defined as the number of slots required for it to reach \emph{all} nodes in the network, after its arrival to the source \texttt{r}. 

\subsection*{Diamond topology}

  We first consider a $4$-node diamond topology as shown  Fig.~\ref{fig:602}. Link capacities are shown along with the links. The broadcast capacity $\lambda^{*}$ of the network is upper bounded by the maximum throughput of node $c$, which is $1$ because at most one of its incoming links can be activated at any time. To show that the broadcast capacity is indeed $\lambda^{*} = 1$, we consider the three spanning trees $\{\mathcal{T}_1, \mathcal{T}_2, \mathcal{T}_3\}$ rooted at the source node $r$. By finding the optimal time-sharing of all feasible link activations over a subset of spanning trees using linear programming, we can show that the maximum broadcast throughput using only the spanning tree $\mathcal{T}_{1}$ is $3/4$. The maximum broadcast throughput over the two trees $\{\mathcal{T}_{1}, \mathcal{T}_{2}\}$ is $6/7$, and that over all three trees $\{\mathcal{T}_1, \mathcal{T}_2, \mathcal{T}_3\}$ is $1$. Thus, the upper bound is achieved and the broadcast capacity is $\lambda^{*}=1$.

We compare our broadcast policy $\pi^{*}$ with the tree-based policy $\pi_{\text{tree}}$ in~\cite{swati}. While the policy $\pi_{\text{tree}}$ is originally proposed to transmit multicast traffic in a wired network by balancing traffic over multiple trees, we slightly modify the policy $\pi_{\text{tree}}$ for broadcasting packets over spanning trees in the wireless setting; link activations are chosen according to the max-weight procedure. See Fig.~\ref{fig:601} for a comparison of the average delay performance under the policy $\pi^{*}$ and the tree-based policy $\pi_{\text{tree}}$ over different subset of trees. The simulation duration is $10^{5}$ slots. We observe that the policy $\pi^{*}$ achieves the broadcast capacity $\lambda^{*}=1$ and is throughput optimal.
\subsection*{Mesh topology}
The broadcast policy $\pi^{*}$ does not rely on the limited tree structures and therefore has the potential to exploit all degrees of freedom in packet forwarding in the network; such freedom may lead to better delay performance as compared to the tree-based policy. To observe this effect, we consider the $10$-node DAG network subject to the primary interference constraint in Fig.~\ref{fig:605}. For every pair of node $\{i,j\}$, $1\leq i <j \leq 10$, the network has a directed link from $i$ to $j$ with capacity $(10-i)$. By induction, we can calculate the number of spanning trees rooted at the source node $1$ to be $9!\approx 3.6\times 10^5$. We choose five arbitrary spanning trees $\{\mathcal{T}_i, 1\leq i \leq 5\}$, over which the tree-based algorithm $\pi_{\text{tree}}$ is simulated. Table~\ref{delay_table} demonstrates the superior delay performance of the broadcast policy $\pi^{*}$, as compared to that of the tree-based algorithm $\pi_{\text{tree}}$ over different subsets of the spanning trees. It also shows that a tree-based algorithm that does not use enough trees would result in degraded throughput.

\begin{figure}
\centering
\subfigure[The wireless network]{
\begin{overpic}[width=0.2\textwidth]{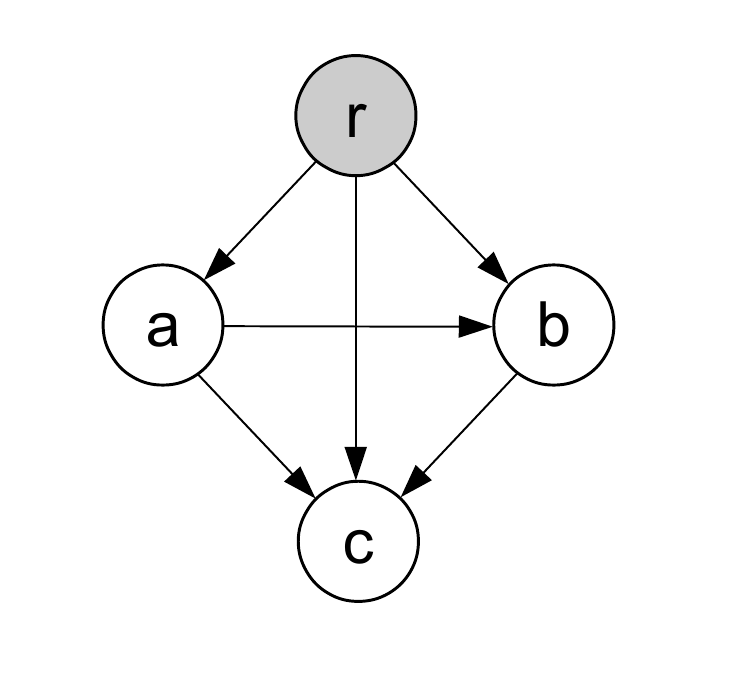}
\put(31,66){\footnotesize $3$}
\put(61,66){\footnotesize $1$}
\put(54,52){\footnotesize $2$}
\put(43,39){\footnotesize $1$}
\put(29,31){\footnotesize $1$}
\put(64,31){\footnotesize $1$}
\end{overpic}
}
\subfigure[Tree $\mathcal{T}_{1}$]{
\includegraphics[width=0.2\textwidth]{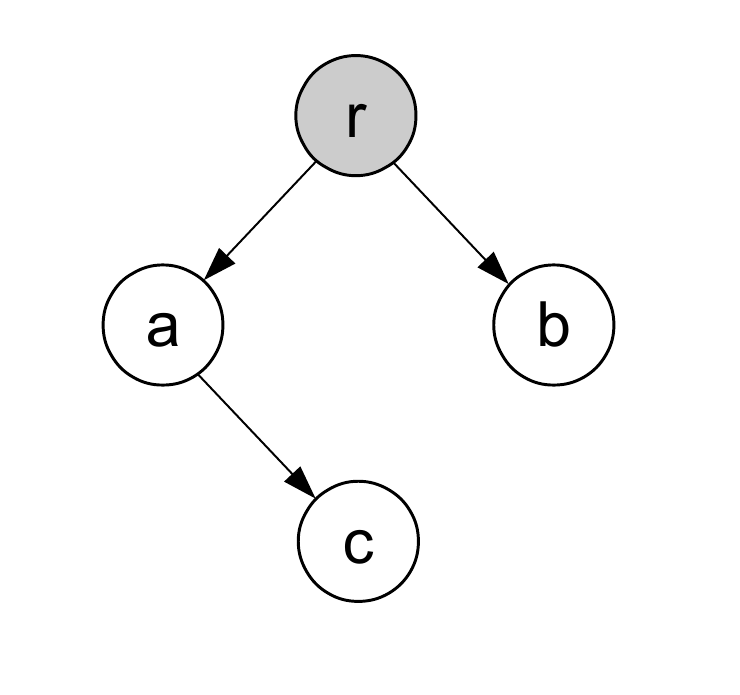}
}
\subfigure[Tree $\mathcal{T}_{2}$]{
\includegraphics[width=0.2\textwidth]{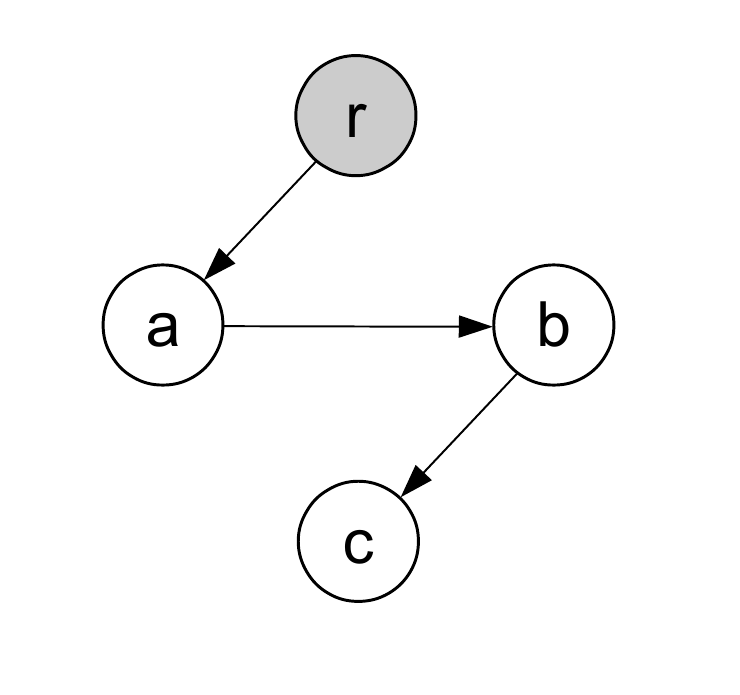}
}
\subfigure[Tree $\mathcal{T}_{3}$]{
\includegraphics[width=0.2\textwidth]{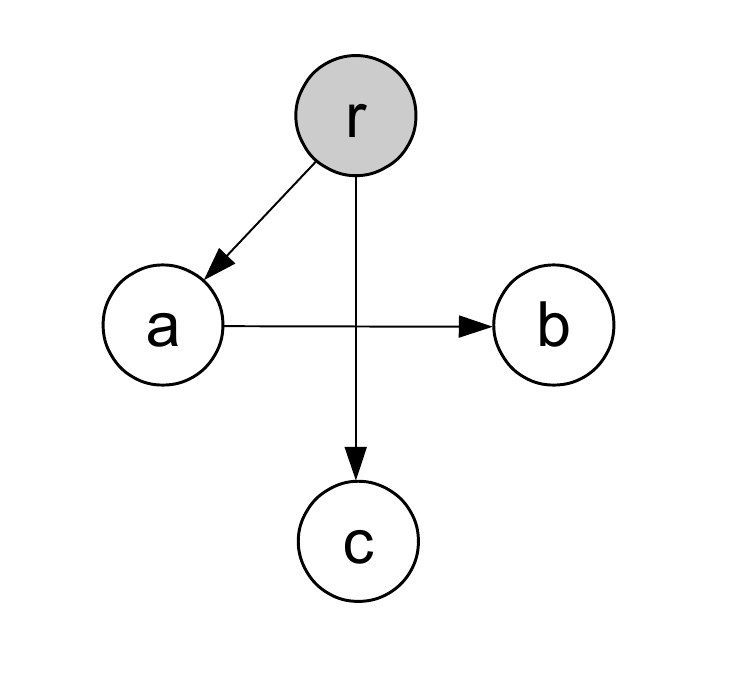}
}
\caption{\small A wireless DAG network and its three embedded spanning trees.}
\label{fig:602}
\end{figure}

%\begin{figure}
%\centering
%\begin{overpic}[width=0.45\textwidth]{Algorithm_Comparisons_ver2}
%\put(30,60){\footnotesize{Optimal Algorithm $(C=1)$}}
%%\put(30,58){\footnotesize{Algorithm $(C=1)$}}
%\put(30,54){\footnotesize{$\mathcal{T}_1,\mathcal{T}_2,\mathcal{T}_3$ $(C=1)$}}
%\put(30,49){\footnotesize{$\mathcal{T}_1,\mathcal{T}_2$ $(C=6/7)$}}
%\put(30,43){\footnotesize{$\mathcal{T}_1$ $(C=3/4)$}}
%\end{overpic}
%%\includegraphics[width=0.5\textwidth]{DAG_Simulations_INFOCOM/Algorithm_Comparisons}
%\caption{\small Average delay performance of the optimal broadcast policy $\pi^{*}$ and the tree-based policy $\pi_{\text{tree}}$ that balances traffic over different subsets of spanning trees.}
%\label{fig:601}
%\end{figure}

\begin{figure}
\centering 
\begin{overpic}[width=0.42\textwidth]{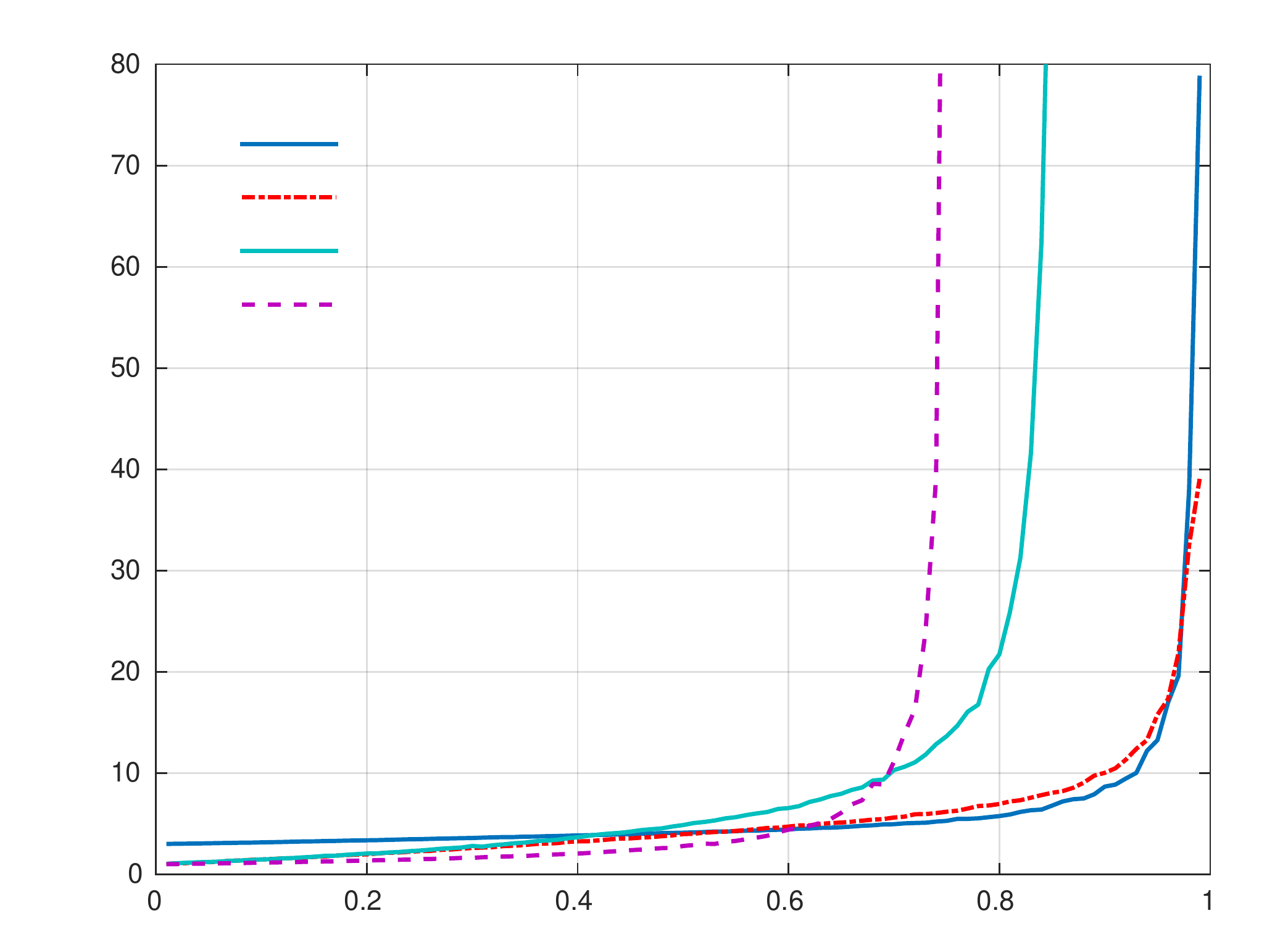}
\put(28,62){\scriptsize{Optimal Algorithm $(C=1)$}}
%\put(30,58){\footnotesize{Algorithm $(C=1)$}}
\put(28,58){\scriptsize{$\mathcal{T}_1,\mathcal{T}_2,\mathcal{T}_3$ $(C=1)$}}
\put(28,54){\scriptsize{$\mathcal{T}_1,\mathcal{T}_2$ $(C=6/7)$}}
\put(28,50){\scriptsize{$\mathcal{T}_1$ $(C=3/4)$}}
\end{overpic}
\caption{\small Average delay performance of the optimal broadcast policy $\pi^{*}$ and the tree-based policy $\pi_{\text{tree}}$ that balances traffic over different subsets of spanning trees.}
\end{figure}

\begin{figure}[ht!]
\centering
\subfigure[The wireless network]{
\includegraphics[scale=0.4]{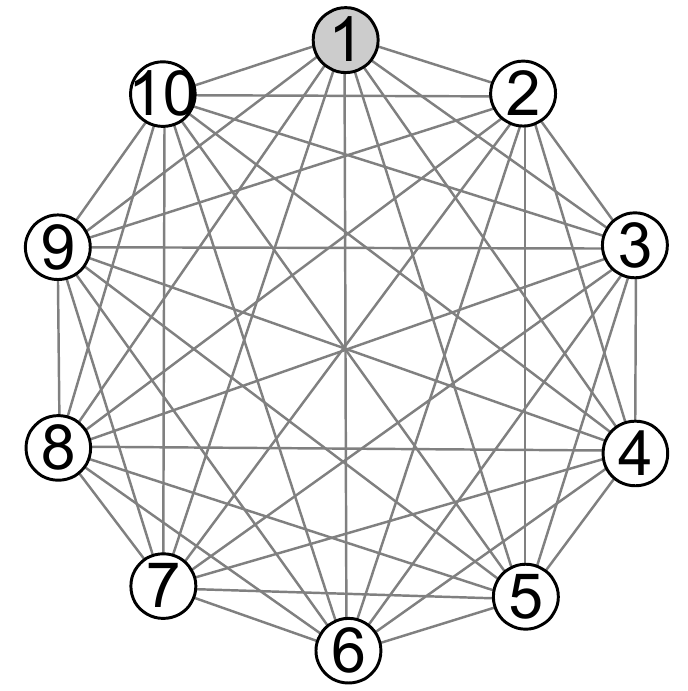}
}
\subfigure[Tree $\mathcal{T}_{1}$]{
\includegraphics[scale=0.4]{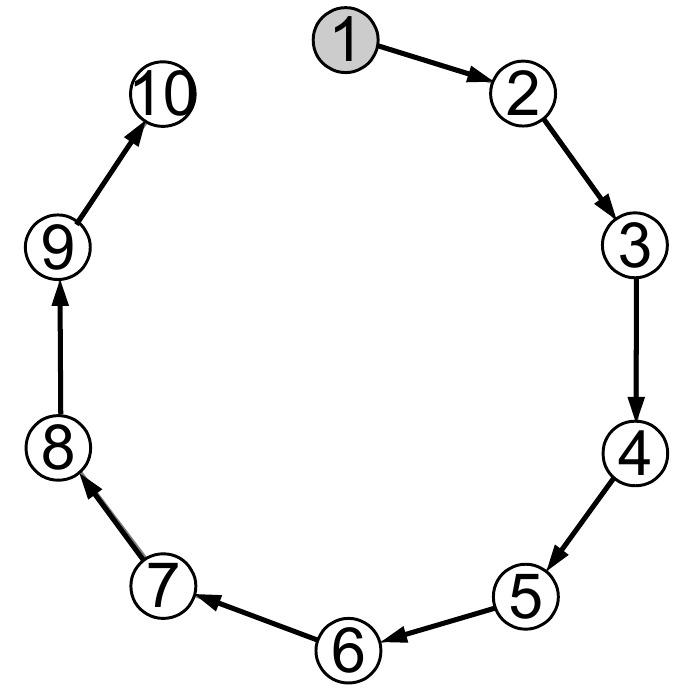}
}
\subfigure[Tree $\mathcal{T}_{2}$]{
\includegraphics[scale=0.4]{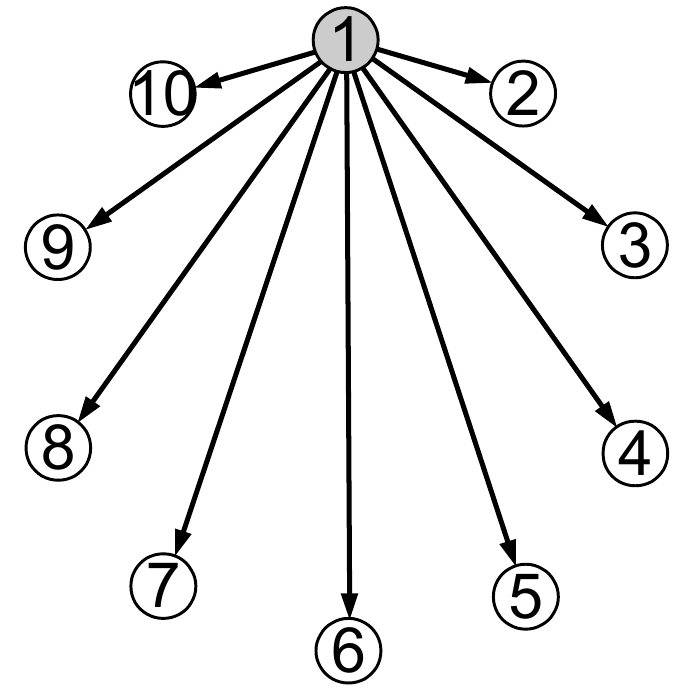}
}
\subfigure[Tree $\mathcal{T}_{3}$]{
\includegraphics[scale=0.4]{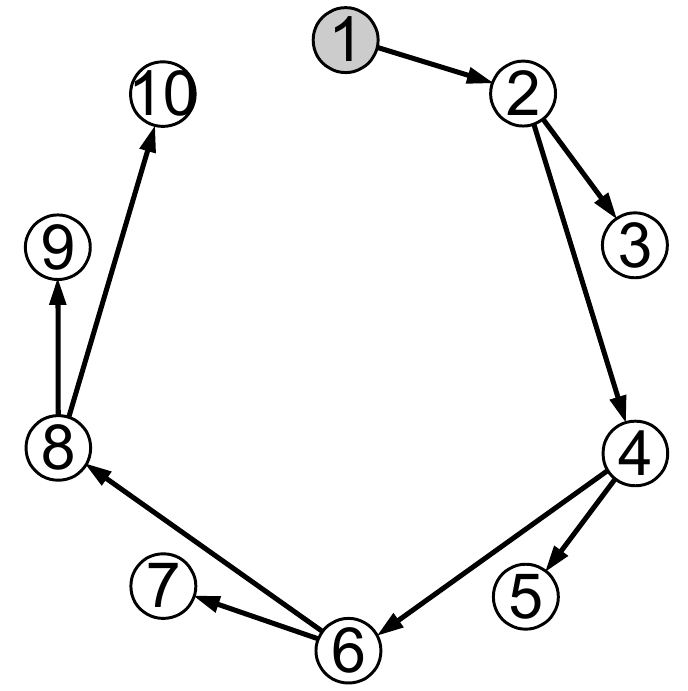}
}
\subfigure[Tree $\mathcal{T}_{4}$]{
\includegraphics[scale=0.4]{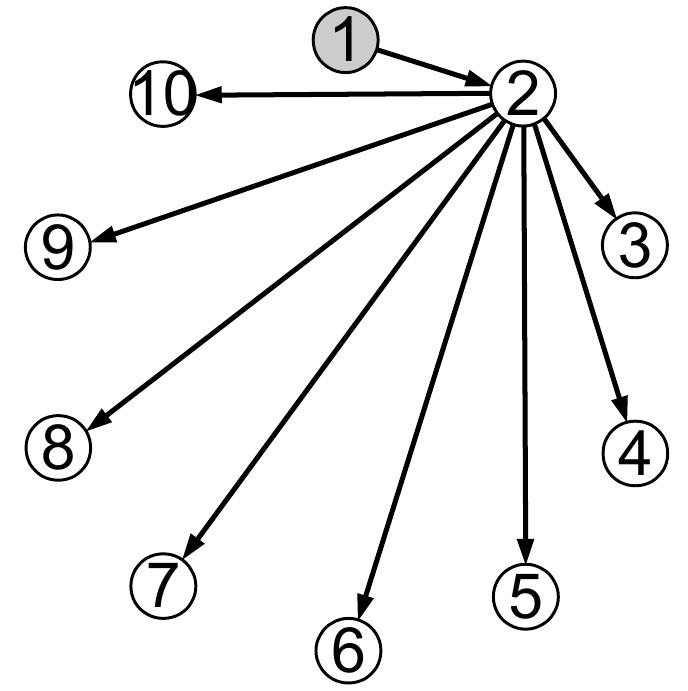}
}
\subfigure[Tree $\mathcal{T}_{5}$]{
\includegraphics[scale=0.4]{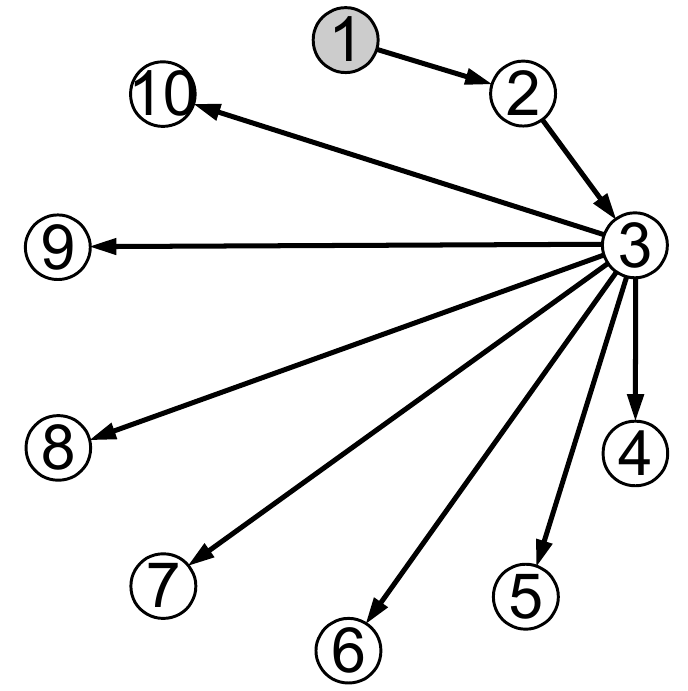}
}
\caption{\small The $10$-node wireless DAG network and a subset of spanning trees.}
\label{fig:605}
\end{figure}

\begin{table*}[ht]
\centering
\begin{tabular}{c | c c c c c | c}
\toprule
& \multicolumn{5}{c|}{tree-based policy $\pi_{\text{tree}}$ over the spanning trees:} & broadcast \\
%\multicolumn{1}{c}{\hspace*{80pt}\textbf{Optimal Algorithm}} \\
%\cmidrule(r){7}
$\lambda$ & $\mathcal{T}_{1}$ & $\mathcal{T}_{1} \sim \mathcal{T}_{2}$ & $\mathcal{T}_{1} \sim \mathcal{T}_{3}$ & $\mathcal{T}_{1} \sim \mathcal{T}_{4}$ & $\mathcal{T}_{1} \sim \mathcal{T}_{5}$ & policy $\pi^{*}$ \\ [0.5 ex]
\hline
$0.5$ & 12.90 & 12.72 & 13.53 & 16.14 & 16.2 & 11.90\\
$0.9$ & $1.3\times 10^4$ & 176.65 & 106.67 & 34.33 & 28.31 & 12.93\\
$1.9$ & $3.31\times 10^4$ & $1.12\times 10^4$ &$4.92\times 10^3$ & 171.56 & 95.76 & 14.67\\
$2.3$ & $3.63\times 10^4$ & $1.89\times 10^4$ & $1.40\times 10^4$& $1.76\times 10^3$ & 143.68 & 17.35\\
$2.7$ & $3.87\times 10^4$& $2.45\times 10^4$ & $2.03\times 10^4$& $1.1\times 10^4$ & 1551.3 & 20.08\\
$3.1$ & $4.03\times 10^4$ & $2.86\times 10^4$& $2.51 \times 10^4$ & $1.78\times 10^4$ & 9788.1 & 50.39\\
\hline
\end{tabular}
\caption{\small Average delay performance of the tree-based policy $\pi_{\text{tree}}$ over different subsets of spanning trees and the optimal broadcast policy $\pi^{*}$.}
\label{delay_table}
\end{table*}
\subsection*{Multiclass Simulation for Arbitrary Topology}
We randomly generate an ensemble of $500$ wired networks (not necessarily DAGs), each consisting of $N=10$ nodes and unit capacity links. By solving the LP corresponding to Eqn. \eqref{multiclass_rate}, we compute the fraction of the total broadcast capacity achievable using $K$ randomly chosen classes by the Multiclass Algorithm \ref{multiclass_algo} of section \ref{cyclic_extension}. The result is presented in Figure \ref{fig:multiclass}. It follows that a sizeable fraction of the optimal capacity may be achieved by using a moderate number of classes. However the number of required classes for achieving a certain fraction of the capacity increases as the broadcast capacity increases. This is because of the fact that increased broadcast capacity would warrant an increased number of DAGs to cover the graph efficiently.
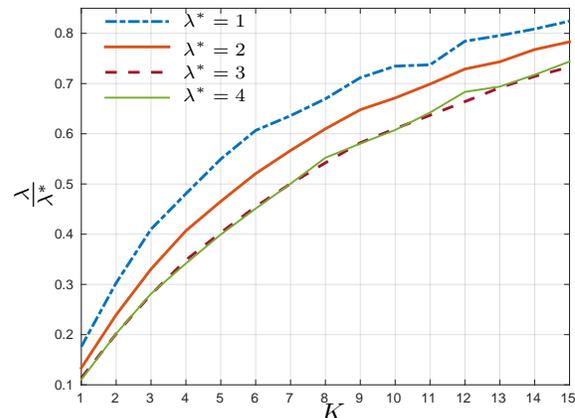
\begin{figure}
\centering
\begin{overpic}[width=0.42\textwidth]{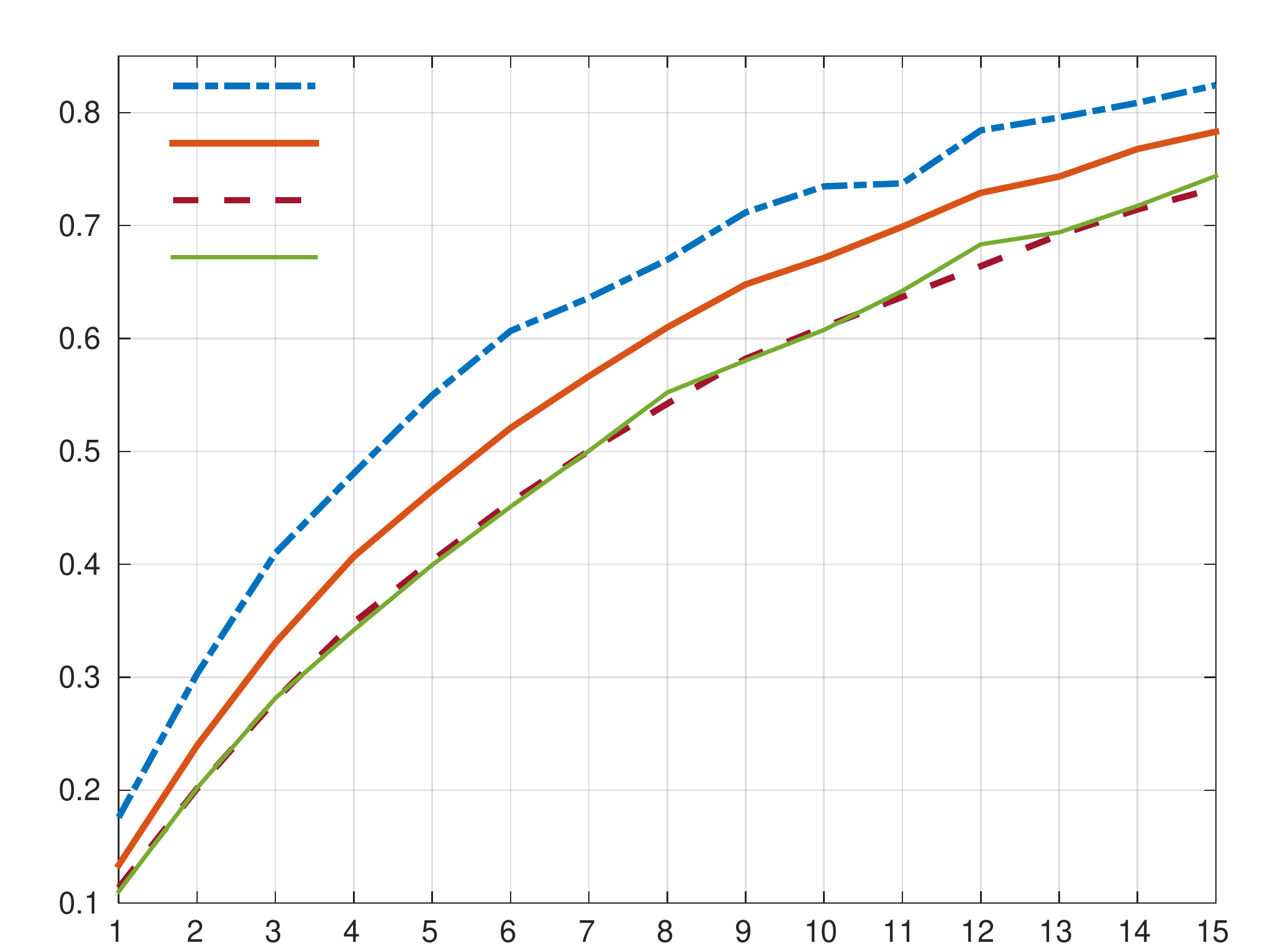}
\put(27,66){\scriptsize{${\lambda^*=1}$}}
\put(27,61){\scriptsize{${\lambda^*=2}$}}
\put(27,57){\scriptsize{${\lambda^*=3}$}}
\put(27,53){\scriptsize{${\lambda^*=4}$}}
\put(-2.5,34){\large{\rotatebox{90}{$\frac{\lambda}{\lambda^*}$}}}
\put(51,-3){$K$}
%\put(30,60){\footnotesize{Optimal Algorithm $(C=1)$}}
%%\put(30,58){\footnotesize{Algorithm $(C=1)$}}
%\put(30,54){\footnotesize{Trees $1,2 \text{ and } 3$ $(C=1)$}}
%\put(30,49){\footnotesize{Trees $1$ and $2$ $(C=6/7)$}}
%\put(30,43){\footnotesize{Tree $1$ $(C=3/4)$}}
\end{overpic}
\caption{\small {Fraction of optimal broadcast rate $\frac{\lambda}{\lambda^*}$ achievable by the multiclass broadcast algorithm with randomly chosen $K$ classes for randomly generated wired networks with $N=10$ nodes.}}
\label{fig:multiclass}
\end{figure}

\section{Conclusion}\label{sec:conclusion}

We characterize the broadcast capacity of a wireless network under general interference constraints. When the underlying network topology is a DAG, we propose a dynamic algorithm that achieves the wireless broadcast capacity. 
Our novel design, based on packet deficits and the in-order packet delivery constraint, is promising for application to other systems with packet replicas, such as multicasting and caching systems. 
Future work involves the study of arbitrary networks, where optimal policies must be sought in the class $\Pi\setminus \Pi^{\text{in-order}}$.

%% file: appendix.tex
\begin{appendix}
%Throughout the proofs in the Appendix, all equalities and inequalities involving random variables hold with probability one unless stated otherwise.

%\subsection{Proof of Lemma \ref{in_expectation_lemma}} \label{in_expectation_lemma_proof}
%Under any policy $\pi \in \Pi$, the number of packets received by a node $i$ up to time $T$ cannot exceed the number of exogenous packets arriving at the source node $r$ up to time $T$. Thus, we have $R_i^{\pi}(T) \leq \sum_{t=1}^{T} A(t)$. Taking a limiting time average at both sides and using the strong law of large numbers on $A(t)$, we get $\limsup_{T\to \infty} \frac{1}{T} R_i^{\pi}(T) \leq \lambda$, which holds for all $i\in V$. It follows that $\max_{i\in V} \limsup_{T\to \infty} \frac{1}{T} R_i^{\pi}(T) \leq \lambda$. If policy $\pi$ is a broadcast policy of rate $\lambda$ (see~\eqref{bcdef}), we have
%\[
%\lambda= \min_{i\in V} \liminf_{T\to \infty}\frac{1}{T} R_i^{\pi}(T)\leq \max_{i\in V} \limsup_{T\to \infty} \frac{1}{T} R_i^{\pi}(T) \leq \lambda.
%\]
%Equivalently, we have for each $i\in V$
%\[
%\lambda= \liminf_{T\to \infty}\frac{1}{T} R_i^{\pi}(T)\leq \limsup_{T\to \infty} \frac{1}{T} R_i^{\pi}(T) \leq \lambda.
%\]
%It follows that the above limits exist and satisfy
%\[
%\lim_{T\to \infty}\frac{1}{T} R_i^{\pi}(T) = \lambda, \ \forall i\in V.
%\]

\subsection{Proof of Theorem \ref{broadcast_ub}} \label{broadcast_ub_proof}

Fix an $\epsilon >0$. Consider a policy $\pi \in \Pi$ that achieves a broadcast rate of at least $\lambda^* -\epsilon$ defined in~\eqref{bcdef}; this policy $\pi$ exists by the definition of the broadcast capacity $\lambda^{*}$ in Definition~\ref{capacity_def}. Consider any proper cut $U$ of the network $\mathcal{G}$. By definition, there exists a node $i \notin U$. Let $\bm{s}^{\pi}(t) = (s_{e}^{\pi}(t), e\in E)$ be the link-activation vector chosen by policy $\pi$ in slot $t$. The maximum number of packets that can be transmitted across the cut $U$ in slot $t$ is at most $\sum_{e\in E_{U}} c_{e} s_{e}^{\pi}(t)$, which is the total capacity of all activated links across $U$, and the link subset $E_{U}$ is given in~\eqref{eq:604}. The number of distinct packets received by a node $i$ by time $T$ is upper bounded by the total available capacity across the cut $U$ up to time $T$, subject to link-activation decisions of policy $\pi$. That is, we have
\begin{equation} \label{bound_packet}
R_i^{\pi}(T) \leq \sum_{t=1}^{T} \sum_{e\in E_{U}} c_{e} s_{e}^{\pi}(t) = \bm{u}\cdot \sum_{t=1}^{T} \bm{s}^{\pi}(t),
\end{equation}
where we define the vector $\bm{u} = (u_{e}, e\in E)$, $u_{e} = c_{e} 1_{[e\in E_{U}]}$, and $\bm{a}\cdot\bm{b}$ is the inner product of two vectors.\footnote{Note that~\eqref{bound_packet} remains valid if network coding operations are allowed.} Dividing both sides by $T$ yields
\[
 \frac{R_i^{\pi}(T)}{T} \leq \bm{u}\cdot\bigg(\frac{1}{T} \sum_{t=1}^{T}\bm{s}^{\pi}(t)\bigg).
\]
It follows that
\begin{align} 
\lambda^*-\epsilon  &\stackrel{(a)}{\leq} \min_{j\in V} \liminf_{T\to \infty} \frac{R_j^{\pi}(T)}{T}\leq  \liminf_{T\to \infty} \frac{R_i^{\pi}(T)}{T} \nonumber \notag  \\
&\leq  \liminf_{T\to \infty}\bm{u}\cdot\bigg(\frac{1}{T} \sum_{t=1}^{T}\bm{s}^{\pi}(t)\bigg) \label{bound2},
\end{align}
where (a) follows  that $\pi$ is a broadcast policy of rate at least $\lambda^*-\epsilon$. Since the above holds for any proper-cut $\bm{u} \in U$, we have 
\begin{eqnarray} \label{bound3}
\lambda^*-\epsilon  \leq \min_{\bm{u} \in U} \liminf_{T\to \infty}\bm{u}\cdot\bigg(\frac{1}{T} \sum_{t=1}^{T}\bm{s}^{\pi}(t)\bigg)
\end{eqnarray}
Now consider the following lemma.

\begin{lemma} \label{conv_hull}
For any policy $\pi \in \Pi$, there exists a vector $\bm{\beta}^{\pi} \in \textrm{conv}(\mathcal{S})$ such that 
\[
\min_{\bm{u} \in U} \liminf_{T\to \infty}  \bm{u}\cdot\bigg(\frac{1}{T} \sum_{t=1}^{T}\bm{s}^{\pi}(t)\bigg) = \min_{\bm{u} \in U}\bm{u}\cdot \bm{\beta}^{\pi} \hspace{10pt} \text{w.p.} 1
\]
\end{lemma}
\begin{IEEEproof}
Consider the sequence $\bm{\zeta}_T^{\pi} = \frac{1}{T}\sum_{t=1}^{T}\bm{s}^{\pi}(t)$ indexed by $T\geq 1$. Since $\bm{s}^{\pi}(t) \in \mathcal{S}$ for all $t\geq 1$, we have $\bm{\zeta}_T^{\pi} \in \conv{\mathcal{S}}$ for all $T\geq 1$. Since $|U|$ is finite, by the definition of $\liminf$, there exists a subsequence $\{\bm{u}\cdot \bm{\zeta}_{T_k}^{\pi}\}_{k\geq 1}$ of the sequence $\{\bm{u}\cdot \bm{\zeta}_{T}^{\pi}\}_{T\geq 1}$ such that
\begin{equation} \label{lim}
\min_{\bm{u} \in U}\lim_{k \to \infty} \bm{u}\cdot \bm{\zeta}_{T_k}^{\pi} = \min_{\bm{u} \in U}\liminf_{T\to \infty} \bm{u} \cdot \bm{\zeta}_T^{\pi}.
\end{equation} 
Since the set $\conv{\mathcal{S}} \subset \mathbb{R}^{|E|}$ is closed and bounded, by the Heine-Borel theorem, it is compact. Hence any sequence in $\text{conv}(\mathcal{S})$ has a converging sub-sequence. Thus, there exists a sub-sub-sequence $\{\bm{\zeta}_{T_{k_i}}^{\pi}\}_{i\geq 1}$ and  $\bm{\beta}^{\pi} \in \text{conv}(\mathcal{S})$ such that 
\[
\bm{\zeta}_{T_{k_i}}^{\pi} \to \bm{\beta}^{\pi}, \quad \text{as $i\to\infty$.}
\]
It follows that
\begin{eqnarray*}
\min_{\bm{u} \in U} \bm{u} \cdot \bm{\beta}^{\pi} &\overset{(a)}{=}& \min_{\bm{u} \in U}\lim_{i \to \infty} \bm{u} \cdot \bm{\zeta}_{T_{k_i}}^{\pi}\\
 &\overset{(b)}{=}& \min_{\bm{u} \in U}\lim_{k \to \infty} \bm{u}\cdot \bm{\zeta}_{T_k}^{\pi} \\
 &\overset{(c)}{=}& \min_{\bm{u} \in U} \liminf_{T\to \infty} \bm{u} \cdot \bm{\zeta}_T^{\pi} \\
&= &\min_{\bm{u} \in U}\liminf_{T\to \infty}  \bm{u}\cdot\bigg(\frac{1}{T} \sum_{t=1}^{T}\bm{s}^{\pi}(t)\bigg),
\end{eqnarray*}
where (a) uses the fact that 
if $\bm{x}_n \to \bm{x}$ then $\bm{c}\cdot \bm{x}_n \to \bm{c}\cdot \bm{x}$ for $\bm{c}$, $\bm{x}_n$, and $\bm{x} \in \mathbb{R}^l$, $l\geq 1$; (b) follows that if the limit of a sequence $\{z_{n}\}$ exists then all subsequences $\{z_{n_k}\}$ converge and $\lim_{k} z_{n_k}= \lim_n z_n$; (c) follows from Equation~\eqref{lim}. This completes the proof of the lemma.
\end{IEEEproof}

Combining Lemma~\ref{conv_hull} with Eqn. ~\eqref{bound3}, we have that there exists a vector $\bm{\beta}^{\pi} \in \text{conv}(\mathcal{S})$ such that
\begin{equation} \label{bdcutC}
\lambda^*-\epsilon \leq \min_{\bm{u} \in U} \bm{u} \cdot \bm{\beta}^{\pi}.
\end{equation} 
Maximizing the right hand side of Eqn. \ref{bdcutC} over all $\bm{\beta}^\pi \in \text{conv}(\mathcal{S})$, we have 
\begin{eqnarray}
\lambda^*-\epsilon \leq \max_{\bm{\beta} \in \text{conv}(\mathcal{S})} \bigg(\min_{\bm{u} \in U} \bm{u} \cdot \bm{\beta} \bigg)
\end{eqnarray}
Since the above inequality holds for any $\epsilon >0$, by taking $\epsilon \searrow 0$ and expanding the dot product, we have 
\begin{eqnarray}
\lambda^* \leq \max_{\bm{\beta} \in \text{conv}(\mathcal{S})} \bigg(\min_{\bm{u} \in \text{$U$: a proper cut}} \sum_{e\in E_{U}} c_{e} \beta_{e} \bigg).
\end{eqnarray}

%Since~\eqref{bdcutC} holds for all proper cuts $U$, we have
%\begin{equation} \label{ub}
% \lambda^* -\epsilon \leq \min_{\text{$U$: a proper cut}} \bm{u} \cdot \bm{\beta}_U^{\pi}.
%\end{equation}
%Lemma~\ref{conv_hull} shows that $ \beta_{U}^{\pi} \in \conv{\mathcal{S}}$ for any proper cut $U$ and policy $\pi \in \Pi$. Thus, from~\eqref{ub} we get
%\begin{equation} \label{conv}
% \lambda^*-\epsilon \leq \sup_{\bm{\beta} \in \conv{\mathcal{S}}} \bigg(\min_{\text{$U$: a proper cut}} \bm{u} \cdot \bm{\beta} \bigg).
%\end{equation}
%The minimum in~\eqref{conv} is the pointwise minimum of a finite number of linear functions, and is concave and continuous \cite{boyd}. Thus, there exists a vector $\bm{\beta}^* \in\conv{\mathcal{S}}$ that achieves the supremum in the compact set $\conv{\mathcal{S}}$~\cite{bertsekas_convex}. Also,~\eqref{conv} holds for all $\epsilon >0$. We conclude that
%\begin{equation} \label{capacity}
%\begin{split}
% \lambda^* &\leq \max_{\bm{\beta} \in \text{conv}(\mathcal{S})} \bigg(\min_{\text{$U$: a proper cut}} \bm{u} \cdot \bm{\beta} \bigg) \\
% &= \max_{\bm{\beta} \in \text{conv}(\mathcal{S})} \min_{\text{$U$: a proper cut}} \sum_{e\in E_{U}} c_{e} \beta_{e}.
% \end{split}
%\end{equation}

\subsection{Proof of Lemma \ref{in_order}} \label{in_order_proof}
Consider the wired network in Fig.~\ref{fig:701}, where all edges have unit capacity and there is no interference constraint. Node $a$ has total incoming capacity equal to two; thus, the broadcast capacity  of the network is upper bounded by $\lambda^{*}\leq 2$. In fact, the network has two edge-disjoint spanning trees as shown in Figures~\ref{fig:702} and~\ref{fig:703}. We can achieve the broadcast capacity $\lambda^{*}=2$ by routing odd-numbered and even-numbered packets along the trees $\mathcal{T}_{1}$ and $\mathcal{T}_{2}$, respectively.
\begin{figure}
\centering
\subfigure[A wired network with a directed cycle $a\to b\to c\to a$.]{
\includegraphics[width=0.17\textwidth]{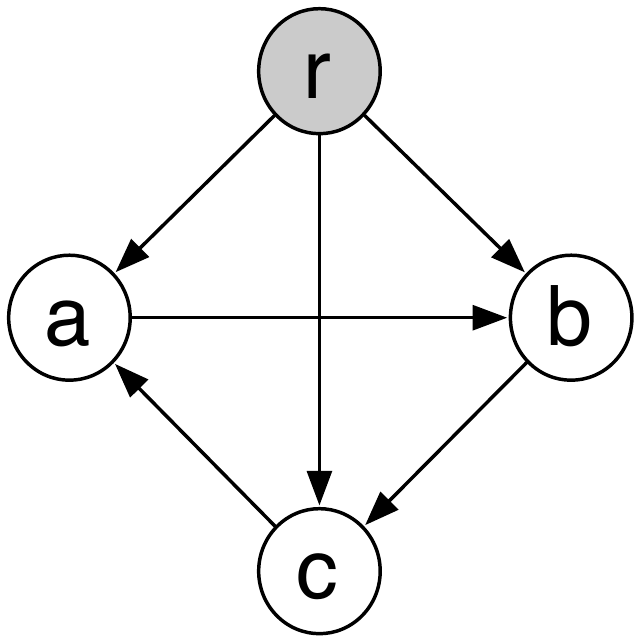}
\label{fig:701}
}
\subfigure[Tree $\mathcal{T}_{1}$]{
\includegraphics[width=0.17\textwidth]{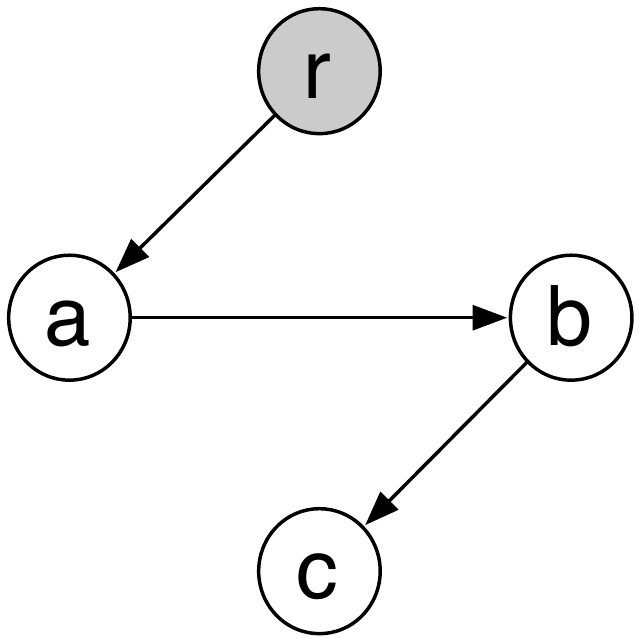}
\label{fig:702}}
\subfigure[Tree $\mathcal{T}_{2}$]{
\includegraphics[width=0.17\textwidth]{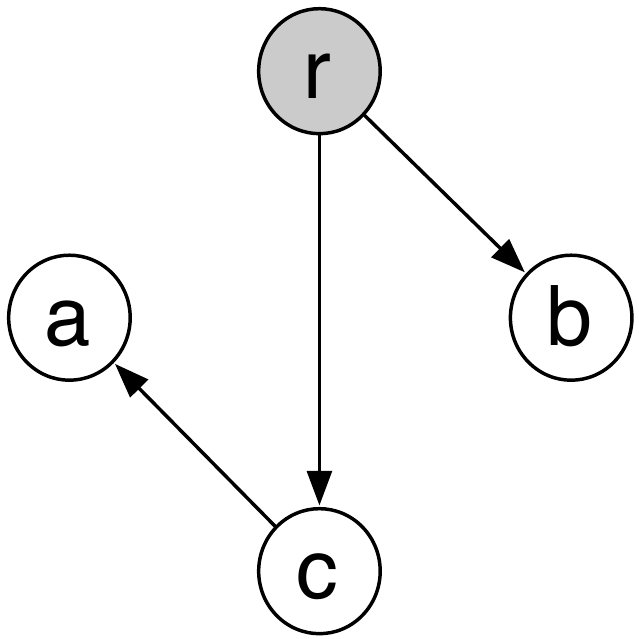}
\label{fig:703}
}
\caption{A wired network and its two edge-disjoint spanning trees that yield the broadcast capacity $\lambda^{*}=2$.}
\label{trees}
% \includegraphics[width=0.32\textwidth]{../Figures/Drawings/Activation2}
  %\captionof{figure}{\textbf{Broadcast Spanning Tree $T_1$}}
 %\includegraphics[width=0.32\textwidth]{../Figures/Drawings/Activation3}
  %\captionof{figure}{\textbf{Broadcast Spanning Tree $T_1$}}
\end{figure}

Consider a policy $\pi \in \Pi_{\text{in-order}}$ that ensures in-order delivery of packets to all nodes. Let $R_i(t)$ be the number of distinct packets received by node $i$ up to time $t$. Hence, node $i$ receives the set of packets $\{1, 2, \ldots, R_{i}(t)\}$ by time $t$ due to in-order packet delivery. Consider the directed cycle $a\to b \to c \to a$ in Fig.~\ref{fig:701}. The necessary condition for all links in the cycle to forward (non-duplicate) packets in slot $t$ is $R_a(t) > R_b(t) > R_c(t) > R_a(t)$, which is infeasible. Thus, there must exist an idle link in the cycle at every slot. Define the indicator variable $x_{e}(t)=1$ if link $e$ is idle in slot $t$ under policy $\pi$, and $x_{e}(t)=0$ otherwise. Since at least one link in the cycle is idle in every slot, we have
\[
 x_{(a, b)}(t) +  x_{(b, c)}(t) +  x_{(c, a)}(t) \geq 1.
\]
Taking a time average of the above inequality yields
\[
\frac{1}{T} \sum_{t=1}^{T} \big(  x_{(a, b)}(t) +  x_{(b, c)}(t) +  x_{(c, a)}(t) \big) \geq 1.
\]
Taking a $\limsup$ at both sides, we obtain
\begin{multline*}
\sum_{e\in\{(a, b), (b, c), (c, a)\}} \limsup_{T\to\infty} \frac{1}{T}\sum_{t=1}^{T} x_{e}(t) \\
\quad \geq \limsup_{T\to\infty} \sum_{e\in\{(a, b), (b, c), (c, a)\}} \frac{1}{T} \sum_{t=1}^{T} x_{e}(t) \geq 1.
\end{multline*}
The above inequality implies that
\begin{equation} \label{inorder_rate}
\max_{e\in\{(a, b), (b, c), (c, a)\}} \limsup_{T\to\infty} \frac{1}{T}\sum_{t=1}^{T} x_{e}(t) \geq \frac{1}{3}.
\end{equation}
Since the nodes $\{a,b,c\}$ are symmetrically located (i.e., the graph obtained by permuting the nodes $\{a, b, c\}$ is isomorphic to the original graph), without any loss of generality, we may assume that the link $e = (a, b)$ attains the maximum in~\eqref{inorder_rate}, i.e.,
\begin{equation} \label{eq:107}
 \limsup_{T\to\infty} \frac{1}{T}\sum_{t=1}^{T} x_{(a, b)}(t) \geq \frac{1}{3}.
\end{equation}
Noting that $x_{e}(t)=1$ if link $e$ is idle in slot $t$ and that node~$b$ receives packets from nodes $r$ and $a$, we can upper bound $R_{b}(T)$ by
\begin{align*}
R_b(T) &\leq \sum_{t=1}^{T} \big(1- x_{(r, b)}(t) +  1-x_{(a, b)}(t)\big) \\
&\leq \sum_{t=1}^{T} \big(2 - x_{(a, b)}(t)\big).
\end{align*}
It follows that
\[
\liminf_{T\to\infty} \frac{R_b(T)}{T} \leq 2 - \limsup_{T\to\infty} \frac{1}{T} \sum_{t=1}^{T} x_{(a, b)}(t) \leq \frac{5}{3},
\]
where the last inequality uses~\eqref{eq:107}. Thus, we have
\[
\min_{i \in V} \liminf_{T\to\infty} \frac{R_{i}(T)}{T} \leq \liminf_{T\to\infty} \frac{R_{b}(T)}{T} \leq \frac{5}{3},
\]
which holds for all policies $\pi\in\Pi_{\text{in-order}}$. Taking the supremum over the policy class $\Pi_{\text{in-order}}$ shows that the broadcast capacity $\lambda^{*}_{\text{in-order}}$ subject to the in-order packet delivery constraint satisfies
\[
\lambda^*_{\text{in-order}} = \sup_{\pi \in \Pi_{\text{in-order}}} \min_{i \in V} \liminf_{T\to\infty} \frac{R_{i}(T)}{T} \leq \frac{5}{3} < 2 = \lambda^*.
\]
i.e., the network broadcast capacity is strictly reduced by in-order packet delivery in the cyclic network in Fig.~\ref{fig:701}.

\subsection{Proof of Theorem~\ref{main_theorem}} \label{main_theorem_proof}

We present the proof in three steps. First, using the dynamics of $X_{j}(t)$ in~\eqref{bnd2}, we derive an expression of the one-slot Lyapunov drift using quadratic Lyapunov functions. Second, we design a useful stationary randomized policy that yields $\epsilon$-optimal broadcast throughput; this policy is useful to show that the system $\bm{X}(t)$, under the optimal broadcast policy $\pi^*$, is strongly stable for all arrival rates $\lambda < \lambda^*$.  Third, based on the above analysis, we show that the policy $\pi^*$ is a throughput-optimal broadcast policy for any underlying network graph which is a DAG.

\begin{lemma} \label{algebra}
If we have
\begin{equation} \label{eq:108}
Q(t+1)\leq  (Q(t)-\mu(t))^+ + A(t) 
\end{equation}
where all the variables are non-negative and $(x)^+ = \max\{x,0\}$, then
\[
Q^2(t+1) - Q^2(t) \leq \mu^2(t) + A^2(t) + 2Q(t)(A(t)-\mu(t)).
\]
\end{lemma}
\begin{IEEEproof}
Squaring both sides of~\eqref{eq:108} yields
\begin{align*}
&Q^2(t+1) \\
&\leq  \big((Q(t)-\mu(t))^+\big)^2 + A^2(t) + 2 A(t)(Q(t)-\mu(t))^+\\
&\leq  (Q(t)-\mu(t))^2 + A^2(t) + 2 A(t)Q(t),
\end{align*}
where we use the fact that $x^2 \geq {(x^+)}^2$, $Q(t) \geq 0$, and $\mu(t) \geq 0$. Rearranging the above inequality finishes the proof.
\end{IEEEproof}
Applying Lemma~\ref{algebra} to the dynamics~\eqref{bnd2} of $X_{j}(t)$ yields, for each node $j\neq r$,
\begin{multline} \label{eq:109}
X_j^2(t+1) - X_j^2(t) \\
\leq  B(t) + 2 X_j(t) \big(\sum_{m\in V}\mu_{mi_{t}^*}(t)-\sum_{k\in V} \mu_{kj}(t)\big),
\end{multline}
where $B(t)\leq \mu^2_{\max}+ \max\{a^2(t),\mu^2_{\max}\} \leq  (a^2(t) + 2\mu^2_{\max})$, $a(t)$ is the number of exogenous packet arrivals in a slot, and $\mu_{\max} \triangleq \max_{e\in E} c_e$ is the maximum capacity of the links.
We assume the arrival process $a(t)$ has bounded second moments; thus, there exists a finite constant $B>0$ such that $\mathbb{E}[B(t)] \leq \mathbb{E}\big(a^2(t)\big) + 2\mu^2_{\max} < B$.

We define the quadratic Lyapunov function $L(\bm{X}(t)) = \sum_{j\neq r} X_j^2(t)$. From~\eqref{eq:109}, the one-slot Lyapunov drift $\Delta(\bm{X}(t))$ satisfies
\begin{align} \label{drift2}
&\Delta(\bm{X}(t)) \triangleq  \mathbb{E}[L(\bm{X}(t+1) - L(\bm{X}(t)) \mid \bm{X}(t)] \notag \\
&= \mathbb{E}\big[\sum_{j\neq r} \big(X_j^2(t+1) - X_j^2(t) \big) \mid \bm{X}(t)\big] \notag \\
&\leq B|V| +2  \sum_{j\neq r} X_{j}(t) \mathbb{E}\big[\sum_{m\in V}\mu_{mi_{t}^*}(t)-\sum_{k\in V} \mu_{kj}(t) \mid \bm{X}(t)\big] \notag \\
&= B|V| - 2 \sum_{(i,j)\in E} \mathbb{E}[\mu_{ij}(t)\mid\bm{X}(t)] \big( X_j(t) - \sum_{k\in K_{j}(t)} X_k(t) \big) \notag \\
&= B|V|- 2 \sum_{(i,j)\in E} \mathbb{E}[\mu_{ij}(t)\mid \bm{X}(t)] \, W_{ij}(t),
\end{align}
where $K_{j}(t)$ and $W_{ij}(t)$ are defined in ~\eqref{eq:110} and~\eqref{eq:111}, respectively. To emphasize that the evaluation of the inequality~\eqref{drift2} depends on a control policy $\pi\in\Pi^{*}$, we rewrite~\eqref{drift2} as
\begin{equation} \label{eq:112}
\Delta^{\pi}(\bm{X}(t)) \leq B|V|- 2 \sum_{(i,j)\in E} \mathbb{E}[\mu^{\pi}_{ij}(t)\mid \bm{X}(t)] \, W_{ij}(t).
\end{equation}
Our optimal broadcast policy $\pi^{*}$ is chosen to minimize the drift on the right-hand side of~\eqref{eq:112} among all policies in $\Pi^{*}$.

Next, we construct a randomized scheduling policy $\pi^{\text{RAND}} \in\Pi^{*}$. Let $\bm{\beta}^*\in\conv{\mathcal{S}}$ be the vector that attains the outer bound on the broadcast capacity $\lambda^{*}$ in Theorem~\ref{broadcast_ub}, i.e.,
\[
\bm{\beta}^* \in \arg \max_{\bm{\beta} \in \conv{\mathcal{S}}} \min_{\text{$U$: a proper cut}} \sum_{e\in E_{U}} c_{e} \beta_{e}.
\]
From Caratheodory's theorem~\cite{matouvsek2002lectures}, there exist at most $(|E|+1)$ link-activation vectors $\bm{s}_l\in \mathcal{S}$ and the associated non-negative scalars $\{p_l\}$ with $\sum_{l=1}^{|E|+1}p_l=1$, such that 
\begin{equation} \label{beta_star}
\bm{\beta}^*= \sum_{l=1}^{|E|+1} p_l \bm{s}_l.
\end{equation}
Hence, from Theorem~\ref{broadcast_ub} we have,
\begin{equation} \label{bc_bound}
\lambda^* \leq \min_{\text{$U$: a proper cut}} \sum_{e\in E_{U}} c_{e} \beta_{e}^{*}.
\end{equation}
Suppose that the exogenous packet arrival rate $\lambda$ is strictly less than the broadcast capacity $\lambda^*$. There exists an $\epsilon >0$ such that $\lambda +\epsilon \leq \lambda^{*}$. From~\eqref{bc_bound}, we have
\begin{equation} \label{eq:113}
\lambda+\epsilon \leq \min_{\text{$U$: a proper cut}} \sum_{e\in E_{U}} c_{e} \beta_{e}^{*}.
\end{equation}
For any network node $v\neq r$, consider the proper cuts $U_{v} = V\setminus \{v\}$. We have, from~\eqref{eq:113}, that
\begin{equation} \label{capacity_exceeding}
\lambda + \epsilon \leq  \sum_{e\in E_{U_{v}}} c_{e} \beta_{e}^{*}, \ \forall v\neq r.
\end{equation}
Since the underlying network topology $\mathcal{G}=(V, E)$ is a DAG, there exists a topological ordering  of the network nodes so that: $(i)$ the nodes can be labelled serially as $\{v_{1}, \ldots, v_{|V|}\}$, where $v_{1}=r$ is the source node with no in-neighbours and $v_{|V|}$ has no outgoing neighbours and $(ii)$ all edges in $E$ are directed from $v_i \to v_j$, $i<j$ ~\cite{algorithms};  From~\eqref{capacity_exceeding}, we define $q_{l}\in[0, 1]$ for each node $v_{l}$ such that 
\begin{equation} \label{q_prob}
q_{l}\, \sum_{e\in E_{U_{v_{l}}}} c_{e} \beta_{e}^{*} = \lambda + \epsilon \frac{l}{|V|},\  l=2, \ldots , |V|.
\end{equation}
Consider the randomized broadcast policy $\pi^{\text{RAND}} \in \Pi^{*}$ working as follows: (i) it selects the feasible link-activation vector $\bm{s}(t) = \bm{s}_{l}$ with probability $p_{l}$ in~\eqref{beta_star}, $l=1,2, \ldots, |E|+1$, in every slot $t$; (ii) for each selected link $e = (\cdot, v_{l})$ of node $v_{l}$ such that $s_{e}(t)=1$, the link $e$ is activated independently with probability $q_{l}$; (iii) activated links are used to forward packets, subject to the constraints that define the policy class $\Pi^{*}$ (i.e., in-order packet delivery and that a network node is only allowed to receive packets that have been received by all of its in-neighbors). Note that this randomized policy is independent of the state $\bm{X}(t)$.  Since each network node $j$ is relabeled as $v_{l}$ for some $l$, from~\eqref{q_prob} we have, for each node $j\neq r$, the total expected incoming transmission rate satisfies
\begin{align} 
\sum_{i: (i, j)\in E}\mathbb{E}[\mu^{\pi^{\text{RAND}}}_{ij}(t)\mid\bm{X}(t)] &=\sum_{i: (i,j)\in E} \mathbb{E}[\mu^{\pi^{\text{RAND}}}_{ij}(t)]  \notag \\
&= q_{l}\, \sum_{e\in E_{U_{v_{l}}}} c_{e} \beta_{e}^{*} \notag \\
&=\lambda + \epsilon \frac{l}{|V|}. \label{rate_comp1}
\end{align}
Equation~\eqref{rate_comp1} shows that the randomized policy $\pi^{\text{RAND}}$ provides each network node $j\neq r$ with the total expected incoming capacity strictly larger than the packet arrival rate $\lambda$ via proper random link activations. According to the abuse of notation in~\eqref{bnd2}, at the source node $r$ we have
\begin{equation} \label{rate_comp2}
\sum_{i:(i,r)\in E} \mathbb{E}[\mu^{\pi^{\text{RAND}}}_{ir}(t)\mid\bm{X}(t)] = \mathbb{E}[\sum_{i:(i,r)\in E} \mu^{\pi^{\text{RAND}}}_{ir}(t)] = \lambda.
\end{equation}
From~\eqref{rate_comp1} and~\eqref{rate_comp2}, if node $i$ appears before node $j$ in the aforementioned topological ordering, i.e., $i = v_{l_{i}} < v_{l_{j}} = j$ for some $l_{i} < l_{j}$, then
\begin{align} 
&\sum_{k:(k,i)\in E}\mathbb{E}[\mu^{\pi^{\text{RAND}}}_{ki}(t)\mid\bm{X}(t)]- \sum_{k:(k,j)\in E}\mathbb{E}[\mu^{\pi^{\text{RAND}}}_{kj}(t)\mid\bm{X}(t)]  \notag \\
&\leq -\frac{\epsilon}{|V|}. \label{rate_comparison_final}
\end{align}
%In this final section of the proof, we will show that the process  $\sum_{(i,j)\in E}Q_{ij}(t)$ is rate stable when driven by the policy $\pi^*$ given in Eqn. \ref{pi_star_def}. \\
The drift inequality~\eqref{drift2} holds for any policy $\pi \in \Pi^*$. Our broadcast policy $\pi^{*}$ observes the system states $\bm{X}(t)$ and seek to minimize the drift at every slot. Comparing the actions taken by the policy $\pi^{*}$ with those by the randomized policy $\pi^{\text{RAND}}$ in slot $t$ in~\eqref{drift2}, we have
\begin{align}
&\Delta^{\pi^*}(\bm{X}(t)) \leq B|V|- 2 \sum_{(i,j)\in E}\mathbb{E}\big[\mu^{\pi^{*}}_{ij}(t) \mid\bm{X}(t)] W_{ij}(t) \notag \\
&\leq B|V|- 2 \sum_{(i,j)\in E}\mathbb{E}\big[\mu^{\pi^{\text{RAND}}}_{ij}(t) \mid\bm{X}(t)] W_{ij}(t) \notag \\
&=  B|V| +2 \sum_{j\neq r} X_j(t) \bigg(\sum_{m\in V}\mathbb{E}\big[\mu^{\pi^{\text{RAND}}}_{mi_{t}^*}(t)\mid\bm{X}(t)\big] \notag \\
&\quad -\sum_{k\in V} \mathbb{E}\big[\mu^{\pi^{\text{RAND}}}_{kj}(t)|\bm{X}(t)\big]\bigg) \notag \\
&\leq B|V| - \frac{2\epsilon}{|V|} \sum_{j\neq r} X_{j}(t). \label{eq:114}
\end{align}
Note that $i_{t}^*= \arg \min_{i\in \text{In}(j)} Q_{ij}(t)$ for a given node $j$. Since node $i_{t}^{*}$ is an in-neighbour of node $j$, $i_{t}^{*}$ must lie before $j$ in any topological ordering of the DAG. Hence, the last inequality of~\eqref{eq:114} follows directly from ~\eqref{rate_comparison_final}. Taking expectation in~\eqref{eq:114} with respect to $\bm{X}(t)$, we have
\[
\mathbb{E}\big[L(\bm{X}(t+1))\big]-\mathbb{E}\big[L(\bm{X}(t))\big] \leq B|V| -\frac{2\epsilon}{|V|}\mathbb{E}||\bm{X}(t)||_1,
\]
where $||\cdot ||_1$ is the $\ell_1$-norm of a vector. Summing the above over $t=0, 1,2,\ldots T-1$ yields
\[
\mathbb{E}\big[L(\bm{X}(T))\big]-\mathbb{E}\big[L(\bm{X}(0))\big] \leq B|V|T -\frac{2\epsilon}{|V|}\sum_{t=0}^{T-1}\mathbb{E}||\bm{X}(t)||_1.
\]
Dividing the above by $2T\epsilon/|V|$ and using $L(\bm{X}(t))\geq 0$, we have 
\begin{eqnarray*}
\frac{1}{T}\sum_{t=0}^{T-1}\mathbb{E}||\bm{X}(t)||_1 \leq \frac{B|V|^2}{2\epsilon} + \frac{|V|\,\mathbb{E}[L(\bm{X}(0))]}{2T\epsilon}
\end{eqnarray*}
Taking a $\limsup$ of both sides yields
\begin{eqnarray} \label{strong_stability}
\limsup_{T \to \infty}\frac{1}{T}\sum_{t=0}^{T-1} \sum_{j\neq r} \mathbb{E}[X_{j}(t)]  \leq \frac{B|V|^2}{2\epsilon}
\end{eqnarray} 
which implies that all virtual-queues $X_{j}(t)$ are strongly stable.

Next, we show that the strong stability of the virtual queues $X_{j}(t)$ implies that the policy $\pi^*$ achieves the broadcast capacity $\lambda^{*}$, i.e., for all arrival rates $\lambda < \lambda^*$, we have 
\[
 \lim_{T\to \infty} \frac{R_j(T)}{T} = \lambda,  \ \forall j.
\]
Equation~\eqref{bnd2} shows that the virtual queues $X_{j}(t)$ have bounded departures (due to the finite link capacities). Thus, strong stability of $X_{j}(t)$ implies that all virtual queues $X_{j}(t)$ are rate stable~\cite[Theorem~$2.8$]{neely2010stochastic}, i.e., $\lim_{T\to \infty} X_j(T)/T = 0, a.s.$ for all $j$. It follows that,
\begin{equation} \label{rate_stability_of_X}
\lim_{T\to \infty} \frac{\sum_{j\neq r}X_j(T)}{T} = 0, \hspace{15pt} \text{w.p.} \hspace{3pt}1
\end{equation}
Now consider any node $j\neq r$ in the network. We can construct a simple path $\sigma (r = u_n \to u_{n-1} \ldots \to u_1 = j)$ from the source node $r$ to the node $j$ by running the following algorithm on the underlying graph $\mathcal{G}(V,E)$. 
 \begin{algorithm} 
\caption{$r\to j$ Path Construction Algorithm}
\begin{algorithmic}[1] 
 \REQUIRE Graph $\mathcal{G}(V,E)$, node $j\in V$
 %\STATE Global variable : $\bm{R}(t)$, Local variables : $\bm{Q}(t),\bm{Q}^{G}(t),\mathcal{A}(t),\bm{X}(t), \bm{W}(t)$.
 \STATE $i \gets 1$
 \STATE $u_i\gets j$
 %\STATE Set $t \gets 1$, $\bm{R}(0)=\bm{0}$.
 \WHILE{$u_i \neq r$} 
 \STATE $u_{i+1} \gets \arg \min_{k\in\text{In}(u_{i})} Q_{ku_i}(t)$; ties are broken arbitrarily.
 \STATE $i \gets i+1$
 \ENDWHILE
 \end{algorithmic}
 \end{algorithm} 
 
This algorithm chooses the parent of a node $u$ in the path $\sigma$ as the one that has the least relative packet deficit as compared to $u$. Since the underlying graph $\mathcal{G}(V,E)$ is a connected DAG (i.e., there is a path from the source to every other node in the network), the above path construction algorithm always terminates with a path $\sigma(r\to j)$. The number of distinct packets received by node $j$ up to time $T$ can be written as a telescoping sum of relative packet deficits along the path $\sigma$, i.e.,
\begin{align}
R_j(T) &= R_{u_1}(T) \notag \\
&= \sum_{i=1}^{n-1}\big(R_{u_i}(T)-R_{u_{i+1}}(T)\big) +R_{u_n}(T)  \notag \\
&= -\sum_{i=1}^{n-1} X_{u_i}(T) + R_r(T) \notag \\
&= -\sum_{i=1}^{n-1} X_{u_i}(T) + \sum_{t=0}^{T-1} A(t), \label{eq:116}
\end{align}
where the third equality follows the observation that (see~\eqref{eq:115})
\[
X_{u_{i}}(T) = Q_{u_{i+1}u_{i}}(T) = R_{u_{i+1}}(T) - R_{u_{i}}(T).
\]
Using $\sum_{i=1}^{n-1} X_{u_{i}}(t) \leq \sum_{j\neq r} X_{j}(t)$,~\eqref{eq:116} and that $X_{j}(t)$ are non-negative, we have, for each node $j$,
\[
\frac{1}{T}\sum_{t=0}^{T-1} A(t) -\frac{1}{T}\sum_{j\neq r} X_{j}(T) \leq \frac{1}{T} R_j(T) \leq \frac{1}{T}\sum_{t=0}^{T-1} A(t).
\]
Taking a limiting time average and the strong law of large numbers for the arrival process, we have
\[
 \lim_{T\to \infty} \frac{R_j(T)}{T} = \lambda, \, \forall j. \hspace{10pt} \text{w.p.}\hspace*{5pt}1
\]
This concludes the proof.

\subsection{Proof of Lemma~\ref{lem:701}} \label{pf:701}
We regard the DAG $G$ as a wired network in which all links can be activated simultaneously. Theorem~\ref{main_theorem} and~\eqref{eq:603} show that the broadcast capacity of the wired network $G$ is
\begin{align}
\lambda^{*} = \lambda_{\text{DAG}} = \min_{\text{$U$: a proper cut}}\, \sum_{e\in E_{U}} c_{e}  &= \min_{\{U_v, v\neq r\} }\, \sum_{e\in E_{U_{v}}} c_{e} \notag \\
&= \min_{v\in V\setminus \{r\}} d_{\text{in}}(v), \label{eq:605}
\end{align}
where $U_{v} = V \setminus \{v\}$ is the proper cut that separates node $v$ from the network, $E_{U_{v}}$ is the set of incoming links of node $v$, and the last equality follows that the maximizer in~\eqref{eq:603} is the all-one vector $\bm{\beta} = \bm{1}$ and that all links have unity capacity. Edmond's Theorem~\cite{edmonds} states that the maximum number of disjoint spanning trees in the directed graph $G$ is
\begin{equation} \label{ed1}
k^* =\min_{\text{$U$: a proper cut}} \sum_{e\in E_{U}} c_e.
\end{equation}
Combining~\eqref{eq:605} and~\eqref{ed1} completes the proof.

\end{appendix}